\documentclass[prx,twocolumn,showpacs,amsmath,amssymb,superscriptaddress,floatfix,reprint,nofootinbib]{revtex4-2}

\usepackage[T1]{fontenc}
\usepackage{amsthm}
\usepackage{booktabs}

\usepackage{macros}

\newcommand{\so}{\mf{so}}
\newcommand{\su}{\mf{su}}
\newcommand{\SO}{\mrm{SO}}
\renewcommand{\O}{\mrm{O}}
\newcommand{\SU}{\mrm{SU}}
\newcommand{\U}{\mrm{U}}
\renewcommand{\S}{\mrm{S}}
\newcommand{\OG}{\mrm{Gr}_0}
\newcommand{\Gr}{\mrm{Gr}}
\renewcommand{\gg}{\hat}
\newcommand{\got}{\mathbin{\gg{\ot}}}
\newcommand{\com}[1]{\mrm{Com}^{(k)}_{#1}}

\makeatletter
\def\l@subsubsection#1#2{}
\makeatother

\makeatletter
\renewcommand\onecolumngrid{
\do@columngrid{one}{\@ne}%
\def\set@footnotewidth{\onecolumngrid}
\def\footnoterule{\kern-6pt\hrule width 1.5in\kern6pt}%
}

\renewcommand\twocolumngrid{
        \def\footnoterule{
        \dimen@\skip\footins\divide\dimen@\thr@@
        \kern-\dimen@\hrule width.5in\kern\dimen@}
        \do@columngrid{mlt}{\tw@}
}%
\makeatother

\begin{document}

\title{Geometry of Free Fermion Commutants}

\author{Marco Lastres}
\email{marco.lastres@tum.de}
\affiliation{Technical University of Munich, TUM School of Natural Sciences, Physics Department, 85748 Garching, Germany}
\affiliation{Munich Center for Quantum Science and Technology (MCQST), Schellingstr. 4, 80799 M\"unchen, Germany}
\author{Sanjay Moudgalya}
\email{sanjay.moudgalya@gmail.com}
\affiliation{Technical University of Munich, TUM School of Natural Sciences, Physics Department, 85748 Garching, Germany}
\affiliation{Munich Center for Quantum Science and Technology (MCQST), Schellingstr. 4, 80799 M\"unchen, Germany}

\date{\today}

\begin{abstract} 
Understanding the structure of operators that commute with $k$ identical replicas of unitary ensembles, also known as their $k$-commutants, is an important problem in quantum many-body physics with deep implications for the late-time behavior of physical quantities such as correlation functions and entanglement entropies under unitary evolution.
In this work, we study the $k$-commutants of free-fermion unitary systems, which are heuristically known to contain $\SO(k)$ and $\SU(k)$ groups without and with particle number conservation respectively, with formal derivations of projectors onto these commutants appearing only very recently.
We establish a complementary perspective by highlighting a larger $\O(2k)$ replica symmetry (or $\SU(2k)$ respectively) that the $k$-commutant transforms irreducibly under, which leads to a simple \textit{geometric} understanding of the commutant in terms of coherent states parametrized by a Grassmannian manifold.
We derive this structure by mapping the $k$-commutant to the ground state of effective ferromagnetic Heisenberg models, analogous to the ones that appear in the noisy circuit literature, which we solve exactly using standard representation theory methods.
Further, we show that the Grassmannian manifold of the $k$-commutant is exactly the manifold of fermionic Gaussian states on $2k$ sites, which reveals a duality between real space and replica space in free-fermion systems.
This geometric understanding also provides a compact projection formula onto the $k$-commutant, based on the resolution of identity for coherent states, which can prove advantageous in analytical calculations of averaged non-linear functionals of Gaussian states, as we demonstrate using some examples for the entanglement entropies.
In all, this work provides a geometric perspective on the $k$-commutant of free-fermions that naturally connects to problems in quantum many-body physics.
\end{abstract}

\maketitle

\section{Introduction} 
A major challenge in physics is to understand and classify the behavior of isolated quantum many-body systems undergoing unitary dynamics.
Generic unitary evolutions are expected to thermalize the system~\cite{d2016quantum}, but the dynamics can become more intricate in the presence of symmetries.
Different classes of dynamics are typically characterized by the late-time behavior of various quantities of physical interest, such as correlation functions and entanglement entropies. 
Given an ensemble of unitaries $\mc U = \{U\}$, the average behavior of simple two-point correlation functions $\langle\widehat{O}(t)^\dagger \widehat{O}(0)\rangle$ under unitary evolution, can be understood using the symmetries of $\mc U$.
On the other hand, the understanding of more complicated quantities such as the entanglement entropies or Out-of-Time-Ordered Correlators (OTOCs)~\cite{Larkin1969, Maldacena2016, Nahum2018, keyserlingk2018hydro}, requires not only the knowledge of the symmetries of $U$, but also of the symmetries of multiple identical \textit{replicas} of the unitary evolution $\mc U^{\otimes k} = \{U^{\ot k}\}$ for some $k \geq 2$~\cite{AmbainisEmerson2007, HarrowLow2009, Gross2007Designs, Dankert2009, brandaoharrowhorodecki2016}.  
The operators commuting with all the unitaries in $\mc U^{\ot k}$ form an associative algebra that is referred to as the \textit{$k$-commutant} of $\mc U$.
In general, the averages of these more complicated quantities over the ensemble $\mc U$ can be expressed in terms of the projectors onto this $k$-commutant.
The $1$-commutant just consists of the symmetries of $\mc U$, which is of course ubiquitously studied for numerous ensembles of unitaries or Hamiltonians. 
These commutants can have a wide variety of structures even for locally generated ensembles of unitaries, which include regular discrete and continuous group-like symmetries, higher-form or non-invertible symmetries~\cite{chatterjee2022algebra, bhardwaj2024lattice}, or exotic symmetries such as quantum many-body scars~\cite{moudgalya2022exhaustive} or Hilbert space fragmentation~\cite{read2007enlarged, moudgalya2022} that result in the breakdown of thermalization~\cite{papic2021review, moudgalya2021review}.
The $2$-commutant, which consist of quadratic symmetries~\cite{zeier2011symmetry, zimboras2015symmetry} or superoperator symmetries~\cite{lastres2026}, naturally appear in the context of dynamical Lie algebras, which have been studied extensively~\cite{Altafini2002, SchulteHerbrueggen2011, d2007introduction, wiersema2023classification}, and are used to characterize the connectivity and universality of unitary operations~\cite{marvian2020locality, lastres2026}.
The $k$-commutant for $k \geq 2$ has also been well-characterized for several ensembles of unitaries, such as the Haar ensemble, where it is exhausted by operators that permute the $k$ identical replicas of the system~\cite{Collins_2006, Dankert2009, Collins_2010, brandaoharrowhorodecki2016}, and has been widely used to study the properties of generic unitary dynamics under random circuits~\cite{fisher2023random}.
A question of particular interest has been up to which value of $k$ a given unitary ensemble possesses the same $k$-commutant as a Haar ensemble which shares the same symmetries.
Such an ensemble of unitaries is said to form a (symmetric) $k$-design~\cite{Gross2007Designs, Dankert2009, HarrowLow2009}.
Unitary ensembles that do not form $k$-designs have a richer $k$-commutant structure, and such a characterization has been recently completed for the Clifford ensemble~\cite{bittel2025clifford}, which is of great importance to resource theories of magic and stabilizer computation, which do not form $k$-designs for $k \geq 4$~\cite{webb2016clifford, zhu2017multiclifford, gross2021clifford}.
It has also been realized that a large class of unitary ensembles that have an on-site symmetry group $G$ also do not form $k$-designs for very large values of $k$, and their $k$-commutants have been completely characterized in recent works ~\cite{hearth2025unitary, liu2024unitary, mitsuhashi2025unitary}.
Another very well-known class of ensembles that do not form $k$-designs for $k \geq 2$ is the Matchgate ensemble built from free-fermion Gaussian unitaries. 
While such unitaries appear to be heuristically known in the literature~\cite{enriched-phases, swann2025, fava2023nonlinear, fava2024} to possess a continuous replica symmetry group $\SO(k)$ or $\SU(k)$ without and with particle number conservation, surprisingly, rigorous characterizations and projection formulas for their $k$-commutants for arbitrary $k$ have only appeared recently~\cite{poetri2026, larocca2026} although they were proven for some low values of $k$ earlier~\cite{Wan2023}.
In this work, our goal is to provide a simple characterization of these $k$-commutants in free-fermion systems, which better highlights the simplicity of their structure, and provide novel methods to compute averages over these free-fermion groups.
In this work, we characterize the free fermion $k$-commutants in terms of ground state manifolds.
It is known that operators in the $1$-commutant, i.e., symmetries of $\mc U$, when viewed as states on $2$ copies of the Hilbert space, can in many cases can be understood as ground states of local frustration-free Hamiltonians that are superoperators that act on $2$ copies of the Hilbert space~\cite{moudgalya2024}.
This understanding also enables the application of conventional intuitions on many-body ground states of frustration-free Hamiltonians to understand the late-time dynamics of physically relevant quantities such as two-point correlation functions, a scenario that also arises in the study of noisy Brownian circuit models~\cite{enriched-phases, ogunnaike2023unifying, moudgalya2024,swann2025,fava2024,Sahu_2024,Agarwal_2022}.
A similar understanding also extends to the $k$-commutants of continuous unitary ensembles, where the operators in the $k$ commutant, when viewed as states on $2k$ copies of Hilbert space, have the form of ground states of appropriate local frustration-free Hamiltonians. 
This understanding was applied to GUE ensemble of unitaries, which shares the same $k$-commutant as the Haar ensemble, to also understand the late-time dynamics of Rényi entanglement entropy from the low-energy excitations of such a Hamiltonian when $k = 2$~\cite{vardhan_entanglement_2024}.
In this work, we derive analogous frustration-free Hamiltonians whose ground state spaces are exactly the $k$-commutants of free-fermion unitaries.
Such Hamiltonians have also naturally appeared in past works in the context of noisy quantum circuit dynamics~\cite{enriched-phases,swann2025,fava2024}.
We then compute said commutants by explicitly showing that these Hamiltonians posses $\SO(2k)$- or $\SU(2k)$-symmetric ferromagnetic ground states, for the free-fermion models without or with particle number conservation respectively.
This has multiple advantages. 
First, it shows that the operators in the $k$-commutant lie in a \textit{single} irreducible representation of the $\O(2k)$ or $\SU(2k)$ group, in contrast to the \textit{multiple} irreducible representations -- whose number grows polynomially with system size -- that are required to characterize the $k$-commutant from its $\SO(k)$ or $\SU(k)$ group structure as done recently~\cite{poetri2026, larocca2026}.
Second, this provides an elegant expression for the projector onto the $k$-commutant, which can now be expressed in terms of integrals over generalized coherent states \cite{generalized-coher-st}, circumventing the need for constructing an orthonormal basis, which would involve a Gelfand-Tsetslin construction~\cite{MOLEV2006109} similar to the one invoked in recent work~\cite{poetri2026, larocca2026}.
A great advantage of the coherent state approach is that these factorize as tensor products of identical states over the whole system, and so the complexity of the integral formula for the projector does not depend on system size.
Finally, our approach leads to a very simple geometric interpretation for the $k$-commutant in terms of the ground state manifold of the effective ferromagnetic model, which we show to be the manifold of fermionic Gaussian states on $2k$ copies of the system, which form the Grassmannian manifolds~\cite{Zirnbauer1996Riemannian, AltlandZirnbauer1997} $\OG(k, 2k$) or $\Gr(k, 2k)$ without or with particle number conservation.

This paper is organized as follows. 
In Sec.~\ref{sec:manycopy}, we introduce the general replica formalism of mapping the $k$-commutants to appropriate ground state manifolds of effective Hamiltonians. 
In Sec.~\ref{sec:mgff}, we apply this formalism to the $k$-commutant of the Matchgate group and derive their manifold structure $\OG(k, 2k)$ that is invariant under $\SO(2k)$. 
In Sec.~\ref{sec:u1ff} we show that the framework is essentially identical for $\U(1)$-symmetric free-fermion unitaries with particle number conservation, where the $k$-commutant now has the manifold structure of $\Gr(k, 2k)$, and is invariant under $\SU(2k)$.
These results can also straightforwardly be extended to free-fermions with various kinds of flavor internal symmetries, as we show in Sec.~\ref{sec:flavorff}.
Finally, in Sec.~\ref{sec:finaltrick}, we demonstrate the advantages of this manifold description by illustrating some computations of physical quantities using generalized coherent states suitable for these commutant manifolds. 
We close in Sec.~\ref{sec:outlook} with a summary of our results and open questions where these techniques might be applicable.
Technical details on these results are relegated to the appendices.

\section{Replica Formalism} 
\label{sec:manycopy}
Consider a many-body tensor product Hilbert space $\mc H_L=\mc H^{\ot L}_1$ and an ensemble of unitaries $\mc U$ that forms a group.
We can define the $k$-commutant of this ensemble $\mc U$ as the algebra of all operators acting on the $k$-fold replicated Hilbert space, that commute with the $k$-th tensor power of all unitaries $U\in\mc U$:
\begin{equation}
	\com{\mc U}\! \defeq\! \{ X\! \in \mathrm{End}(\mc H_L^{\ot k}) : [U^{\ot k}\!,X] = 0, \forall\, U\!\in \mc U \}.\label{eq:comm}
\end{equation}
In most examples we study here, the ensemble $\mc U$ will be the complete set of unitaries obtained by composing local unitaries generated by 2-site Hermitian interactions $\mc G=\{h_{\alpha,ij}\}$ acting on pairs of adjacent sites, denoted $\langle ij\rangle$; such a set is simply the exponential of the Lie algebra generated by the operators $\mc G$~\cite{d2007introduction}, i.e. the \textit{dynamical Lie algebra}:
\begin{equation}
	U=\exp(iH),\quad H\in \mf{Lie}(\mc G).
\end{equation}
In such cases, the $k$-commutant is equivalently the commutant of the $k$-fold replicated generators
\begin{equation}
	h_{\alpha,ij}^{(k)}\defeq\sum_{l=1}^k \1^{\ot (l-1)}\ot h_{\alpha,ij}\ot\1^{\ot (k-l)},
\end{equation}
which generate the tensor product representation of $\mc U$.
\subsection{Effective Hamiltonians}
Since the constraints imposed on $X\in\com{\mc U}$ by Eq.~\eqref{eq:comm} are linear, the problem of characterizing the $k$-commutant can equivalently be expressed by thinking of (super)operators $X\in \mathrm{End}(\mc H_L^{\ot k})$ as states ${\ket{X}\in\mc H_L^{\ot 2k}}$ on twice the number of copies of the Hilbert space.
For example, when $k=1$:
\begin{equation}
	X = \sum_{\mu\nu} X_{\mu\nu} \ket{\mu}\!\bra{\nu} \iff \ket X = \sum_{\mu\nu} X_{\mu\nu} \ket{\mu}\ket{\nu}.\label{eq:vectorization}
\end{equation}
For general $k$, we use the convention that odd copies of the Hilbert space correspond to `ket' states $\ket\psi\in\mc H_L$, and even copies to `bra' states in the dual $\bra\psi\in\mc H_L^*$, where $\mc H_L$ and $\mc H_L^\ast$ are related by complex conjuga\-tion/time reversal; hence the mapping is base-dependent.
For a unitary operator $U$, in this representation we have
\begin{equation}
    [U^{\ot k},X]=0 \iff (U \otimes U^\ast)^{\ot k}\ket X=\ket X,\label{eq:unitary-comm}
\end{equation}
so that $\com{\mc U}$ is the set of states ${\ket{X}\in\mc H_L^{\ot 2k}}$ which are stabilized by $(U \otimes U^\ast)^{\ot k}$ for all $U\in\mc U$.
When $\mc U$ is generated by the Lie algebra $\mf{Lie}(\mc G)$, we can 
also express the commutant in terms of the action of commutators of $\mc G$ on $X$.
These can be expressed through the adjoint ``Liouvillian'' operators as:
\begin{gather}
	\mc L^{(k)}_{\alpha,ij}\!\defeq\!\sum_{l=1}^k\1^{\ot 2(l-1)}\ot(h_{\alpha,ij}\ot\1-\1\ot h_{\alpha,ij}^*)\ot\1^{\ot 2(k-l)}\nn \\
    \mc L^{(k)}_{\alpha,ij}\ket X=|\,[h_{\alpha,ij}^{(k)},X]\,\rangle.
\end{gather}
In that case, we can equivalently formulate Eq.~\eqref{eq:comm} as:
\begin{equation}
\begin{aligned}
	\com{\mc U}\! &=\! \{ \ket X\! \in \mc H_L^{\ot 2k} : \mc L^{(k)}_{\alpha,ij}\ket X = 0, \forall\ h_{\alpha,ij}\!\in \mc G \}\\
	&=\bigcap_{\alpha,\langle ij\rangle}\ker\big(\mc L^{(k)}_{\alpha,ij}\big).\label{eq:comm-ad}
\end{aligned}
\end{equation}
We can then also express the commutant as the ground state space of an effective positive semi-definite Hermitian Hamiltonian~\cite{moudgalya2024}
\begin{equation}
	P^{(k)}=\sum_{\langle ij\rangle} P_{ij}^{(k)},\quad P_{ij}^{(k)} = \sum_\alpha \big(\mc L_{\alpha,ij}^{(k)}\big)^2,
\end{equation}
so that
\begin{equation}
    \com{\mc U}=\ker(P^{(k)})=\bigcap_{\langle ij\rangle}\ker(P_{ij}^{(k)}).
\end{equation}
Note that since the generators $\mc G$ are spatially local, the Hamiltonian $P^{(k)}$ is also local in the spatial direction, although each term $P_{ij}^{(k)}$ couples all copies with each other.
Also note that the symmetries of the effective Hamiltonian $P^{(k)}$ must not be conflated with the symmetries of the replicated unitary operators $U^{\ot k}$, i.e. the elements of the $k$-commutant.
In this picture, the $k$-commutant is the ground state space of $P^{(k)}$, and the symmetries of $P^{(k)}$ -- which we call \textit{replica symmetries} -- will be useful for characterizing it.
Such effective Hamiltonians, sometimes referred to as super-Hamiltonians in the literature, naturally appear in the analysis of the dynamics of various quantities pertaining to operator and entanglement dynamics under noisy Brownian circuit models, and they have been studied in those contexts \cite{enriched-phases,swann2025,fava2024, Zhang2022universal, moudgalya2024, vardhan_entanglement_2024, lastres2026, PhysRevX.15.021020, swann2026continuummechanicsentanglementnoisy, Bernard_2021, 10.21468/SciPostPhys.15.4.175}.

\subsection{Ground State Manifolds}
While identifying specific or even most elements of the $k$-com\-mu\-tant $\com{\mc U}$ is generally straightforward, establishing the \textit{exhaustive} set of such operators can often be extremely challenging.
Here we show that a significant simplification occurs if the effective Hamiltonian $P^{(k)}$ can be proven to be \textit{ferromagnetic}, possibly after a suitable redefinition of the many-body basis.
This condition is defined as
\begin{equation}
    \com{\mc U} = \ker(P^{(k)}) = \text{span}\{\ket{v}^{\ot L}\},\;\;\ket{v}\in\mc H_1^{\ot 2k}. 
\end{equation}
For example, for the $k = 1$ commutant, on a qubit Hilbert space ($\mc H_1 = \mathbb{C}^2$), we have that~\cite{moudgalya2022, moudgalya2022from,  moudgalya2024}
\begin{equation}
    \mrm{Com}^{(1)}_{{\mc U}_{\mb Z_2}} = \text{span}\{\mathds{1}, \Gamma \} = \text{span}\{\ket{I}^{\ot L}, \ket{Z}^{\ot L}\},
    \label{eq:Z2comm}
\end{equation}
where ${\mc U}_{\mb Z_2}$ is the set of unitaries symmetric under the $\mb Z_2$ symmetry operator $\Gamma = \prod_{j=1}^L{Z_j}$, and $\ket{I}, \ket{Z} \in \mH_1^{\ot 2}$ are the vectorizations of the on-site identity $I$ and Pauli $Z$ operators.
Similarly, for $\U(1)$-symmetric operators ${\mc U}_{\U(1)}$, with charge $Z_\mrm{tot}=\sum_{j=1}^L Z_j$, we have~\cite{moudgalya2022, moudgalya2022from, ogunnaike2023unifying,  moudgalya2024}
\begin{equation}
\begin{split}
    \mrm{Com}^{(1)}_{{\mc U}_{\U(1)}} = \text{span}_{n = 0}^{L-1}\{Z_\mrm{tot}^n\} \qquad\qquad\qquad\qquad \\
    \qquad = \text{span}_{\theta, \varphi}\Big\{\Big(\cos\frac{\theta}{2}\ket{I} + e^{i\varphi} \sin\frac{\theta}{2}\ket{Z}\Big)^{\ot L}\Big\}.
\label{eq:U1comm}
\end{split}
\end{equation}
Moreover, for general $k$, the $k$-commutants of the Haar or GUE ensembles of unitary operators without any symmetry, are also of the ferromagnetic form~\cite{PhysRevX.7.031016, Nahum2018, hunterjones2019unitary, hearth2025unitary, vardhan_entanglement_2024}
\begin{equation}
    \com{\mc U_{\rm Haar}} = \text{span}_{\sigma \in S_k}\{\ket{\sigma}^{\ot L}\},
\label{eq:Haarcommutant}
\end{equation}
where each $\ket\sigma$ is a state on $2k$ copies of the local Hilbert space $\mc H_1^{\ot 2k}$, and corresponds to an element of the permutation group $S_k$.
For ferromagnetic Hamiltonians, the ground state space can be fully characterized by a \textit{ground state manifold} given by
\begin{equation}
    M_{\mc U}^{(k)}=\mb P\big\{\!\ket{v}\in \mc H_1^{\ot 2k}:P_{1,2}^{(k)}(\ket{v}\ot\ket{v})=0\big\},
\end{equation}
where $\mb P\{\cdot\}$ denotes the projectivization operation (the condition $P_{1,2}^{(k)}(\ket{v}\ot\ket{v})=0$ is invariant under scalar multiplication $\ket{v}\mapsto\lambda\ket{v}$, so we identify valid local states up to an overall rescaling).
Thus $M_{\mc U}^{(k)}$ is a well-defined compact manifold in the complex projective space $\mathbb{P}(\mc H_1^{\ot 2k})$, independent of system size $L$.
For example, from Eq.~\eqref{eq:Z2comm} and Eq.~\eqref{eq:U1comm}, it is clear that $M_{\mc U_{\mb Z_2}}^{(1)}$ is a discrete manifold with two points, and that $M_{\mc U_{\U(1)}}^{(1)}$ is the Bloch sphere:
\begin{equation}
    M_{\mc U_{U(1)}}^{(1)} \cong \mathbb{CP}^1 = \mb S^2.
\end{equation}
As we will discuss in a separate upcoming work~\cite{lastres2026-2}, the ground state manifold is not always a \textit{manifold} in the mathematical sense, since it can possess singularities, as in $M^{(k)}_{{\mc U}_{\U(1)}}$ for ${k\geq 2}$.
However, as we discuss below, the continuous replica symmetry present in the free-fermion unitary ensembles studied here ensures that their ground state manifolds are smooth and homogeneous.
This homogeneity property will also guarantee that $M_{\mc U}^{(k)}$ is a manifold of generalized coherent states \cite{generalized-coher-st}.
The manifold $M^{(k)}_{\mc U}$ should not be confused with the Lie group manifold that the symmetry generators of $\com{\mc U}$ might form.
In general, $M^{(k)}_{\mc U}$ might be equal or larger.
For instance, in the $\mb Z_2$ example above the two coincide, while in the $\U(1)$ example, the group manifold is a circle $\U(1)\cong\mb S^1$, while the ground state manifold associated to the commutant is a sphere $M_{\mc U_{U(1)}}^{(1)} \cong \mb S^2$.
\subsubsection*{Proving ferromagnetism}
There are various available strategies which allow to prove that a given Hamiltonian has a frustration-free ferromagnetic ground state space.
In this work we will focus on the use of tools from representation theory, which we will discuss in the following sections, and in more detail in Appendix~\ref{app:repr}.
Let us however preview a simple general result on these states.
An important first step towards proving that a Hamiltonian is ferromagnetic is proving that its ground states are invariant under the permutation of the $L$ physical sites, since any state spanned by fully-polarized states of the form $\ket v^{\ot L}$ has this symmetry.
It is easy to show that this symmetry property need only to be verified locally, resulting in the following simple lemma, which we prove in Appendix~\ref{app:proof}.
\begin{lemma}\label{lem}
	Let $\{m_{\alpha}\}$ be a set of two-site operators on $\mc H_\mrm{loc}^{\ot 2}$, and let $m_{\alpha,ij}$ be the operator $m_\alpha$ acting on a pair of adjacent sites $\langle ij\rangle$ in a many-body Hilbert space $\mc H_\mrm{loc}^{\ot L}$.
    Suppose the intersection of the kernels of the operators $\{m_\alpha\}$ only contains states symmetric under exchange:
	\begin{equation}
		\bigcap_\alpha \ker(m_\alpha) \subseteq \mrm{Sym}^2(\mc H_\mrm{loc}).\label{eq:sym-condition}
	\end{equation}
	Then, if adjacent sites $\langle ij\rangle$ form a connected graph, we also have that the states in the common kernel of all $m_{\alpha, ij}$ are globally symmetric, and
	\begin{equation}
		\bigcap_{\alpha,\langle ij\rangle} \!\!\ker(m_{\alpha,ij}) \!=\! \{\ket\psi\in\mrm{Sym}^L(\mc H_\mrm{loc}):m_{\alpha,\tilde\imath\tilde\jmath}\ket\psi=0,\forall\alpha\},
	\end{equation}
	where $\tilde\imath\tilde\jmath$ are any two sites.
\end{lemma}
As a consequence, if condition \eqref{eq:sym-condition} is satisfied, the ground state space will be independent of the geometry of the lattice bonds $\langle ij\rangle$, provided the lattice is fully connected.
Problems on one-dimensional spin-chains can be reduced to problems with all-to-all interactions and vice versa. 
We will always assume that the bonds $\langle ij\rangle$ in our systems form a connected graph.

\begin{table*}[t] 
\centering
\renewcommand{\arraystretch}{1.4}
\begin{tabular}{c c c c c }
\toprule
\textbf{System} & \textbf{Replica sym.} & \ \ \textbf{$\bm{\dim(\com{\mc U})}$}\ \  & \textbf{Manifold ($\bm{M_{\mc U}^{(k)}}$)} & \textbf{$\bm{\dim(M_{\mc U}^{(k)})}$}\\
\midrule
Matchgates/Majorana F.F. & $\O(2k)$ & $2\times$Eq.~\eqref{eq:maj-dim} & $\OG(k,2k)=\O(2k)/\U(k)$ & $k(k-1)$ \\
+ parity-flip & $\SO(2k)$ & Eq.~\eqref{eq:maj-dim} & $\OG^+(k,2k)=\SO(2k)/\U(k)$ & $k(k-1)$ \\
\midrule
$\SO(N)$ Majorana F.F. & \multicolumn{4}{c}{\scriptsize{(same as the $kN$-th commutant Majorana free fermions)}} \\
\midrule
$\U(1)$ Free Fermions & $\SU(2k)$ & Eq.~\eqref{eq:dim-u1} & $\Gr(k,2k)=\U(2k)/(\U(k)\times \U(k))$ & $2k^2$ \\
\midrule
$\SU(N)$ Free Fermions &  \multicolumn{4}{c}{\scriptsize{(same as the $kN$-th commutant of $\U(1)$ free fermions)}} \\
\bottomrule
\end{tabular}
\caption{Structure of the $k$-commutants studied in this work. In the main text we discuss the connection between the ground state manifolds $M^{(k)}_\mrm{\mc U}$ and the manifold of fermionic Gaussian states, while in Appendix~\ref{app:repr} we build them algebraically.
The replica symmetry is the symmetry group of the effective Hamiltonian, and of the manifold of the $k$-commutant.
This group should not be confused with the $k$-commutant itself, which does contain an $\SO(k)$ or $\SU(k)$ group in the cases without and with particle number conservation~\cite{enriched-phases,swann2025,  fava2024}, but actually forms an associative algebra of operators, which we describe as a single irrep of the replica symmetry group.
Dimensions of the manifolds are given over the field of real numbers.\label{tab:manifolds}}
\end{table*}

\section{Matchgates (Majorana Free Fermions)} 
\label{sec:mgff}
Let us consider a system of $2L$ Majorana modes $\{\gamma_i\}_{i=1}^{2L}$ which satisfy $\gamma\+_i=\gamma_i$ and $\{\gamma_i,\gamma_j\}=2\delta_{ij}\1$. The most general quadratic Hamiltonian takes the form $\frac{i}{2}\sum_{\langle ij\rangle }K_{ij}\gamma_i\gamma_j$, where $K_{ij}$ is a real antisymmetric matrix, so that local generators are 
\begin{equation}
h_{ij}=i\gamma_i\gamma_j.
\end{equation}
These generate the group of fermionic Gaussian unitaries, i.e., the Matchgate group $\mc U_\mrm{MG}\cong \SO(2L)$~\cite{free-fermion-group,matchgateso2n, klebanov2018spectra, pakrouski2020many}. 
\subsection{Effective Hamiltonian}\label{subsec:Majoranaeffectiveham}
Let us now consider the system with $2k$ copies of the Hilbert space.
If we label the copies using the index $a\in\{1,...,2k\}$, we find\footnote{Note that there is not a unique way to deal with the replicas, and the conventions we use differ from the ones in the Refs.~\cite{enriched-phases} and \cite{swann2025}. See also App.~\ref{app:conventions} for a discussion on this.}
\begin{equation}
	\mc L^{(k)}_{ij} \!= i\sum_{a=1}^{2k} \tilde\gamma_i^a\tilde\gamma_j^a,\ \ \text{where }\tilde\gamma_i^a=\begin{cases}
		\gamma_i^a \!\!\!\!&\text{if }a\text{ is odd,}\\
		(\gamma_i^a)^* \!\!\!\!&\text{if }a\text{ is even.}
	\end{cases}\label{eq:majorana-ad}
\end{equation}
For any $a$, the set of operators $\{\tilde\gamma_i^a\}_{i=1}^{2L}$ forms a valid set of Majorana modes, but notice that for regular tensor products, modes on different copies commute instead of anti-commuting.
As we discuss in App.~\ref{app:conventions}, this is not the only possible convention, but it is most natural when the $\{\gamma_i\}_{i=1}^{2L}$ correspond to the fermionization of a spin chain through a Jordan-Wigner transformation.
With this convention, the modes can be made to anti-commute using Klein factors:
\begin{equation}
	\bar\gamma_{i}^a = \left(\prod_{b=1}^{a-1}\tilde\Gamma^b\right)\tilde\gamma_i^a,\quad \tilde\Gamma^b=(-i)^L\prod_{j=1}^{2L}\tilde\gamma_j^b,\label{eq:maj-kelin}
\end{equation}
where $\{\tilde\Gamma^b,\tilde\gamma_i^b\}=0$ and $(\tilde\Gamma^b)^2=1$.
In terms of $\bar\gamma_i^a$, the ex\-pression for $\mc L_{ij}^{(k)}$ does not change, and we obtain~\cite{enriched-phases,swann2025}
\begin{equation}
\begin{gathered}
	P_{ij}^{(k)}=\big(\mc L_{ij}^{(k)}\big)^2 =2k-8\sum_{a<b}^{2k}J^{ab}_iJ^{ab}_j,\\
	\mc L_{ij}^{(k)} = i\sum_{a = 1}^{2k}{\bar \gamma_i^a \bar\gamma_j^a},\quad J^{ab}_i\defn-\frac{i}{2}\bar\gamma_i^a\bar\gamma_i^b. \label{eq:Jabdefn}
\end{gathered}
\end{equation}
where $J^{ab}_i$ are the local $\so(2k)$ generators, which commute on different sites and satisfy the relations
\begin{equation}
	[J_i^{ab}, J_i^{cd}] = i(\delta_{ac} J_i^{bd}+\delta_{bd} J_i^{ac}-\delta_{ad} J_i^{bc}-\delta_{bc} J_i^{ad}).\label{eq:gen-so2k}
\end{equation}
Note that the local chirality operators
\begin{equation}
    \bar\Gamma_i \defeq (-i)^k \prod_{a=1}^{2k} \bar{\gamma}_i^a
\end{equation}
commute with $P^{(k)}$, hence the $2^k$-dimensional local Fock space $\mc H_\mrm{loc}$ on each replicated site splits into two $2^{k-1}$-dimensional subspaces of opposite chirality, which transform as irreducible $\so(2k)$ spinor representations.
To understand the frustration-free ground state of $P^{(k)}$, it is convenient to express the interaction in terms of the $\so(2k)$ Casimir operator
\begin{equation}
	C \defeq \sum_{a<b}^{2k} (J^{ab})^2.
\end{equation}
The Casimir operator is defined for any $\so(2k)$ representation: for the on-site representation we define $C_i$ by choosing the local generators $J^{ab}=J^{ab}_i$, and for a pair of sites we define $C_{ij}$ by setting $J^{ab}=J^{ab}_i+J^{ab}_j$.
We also define the global Casimir operator $C_{\mrm{tot}}$ where $J^{ab}=\sum_i J^{ab}_i$.
In this way we find $P_{ij}^{(k)}=2k+4(C_i+C_j-C_{ij})$.
It is easy to see using Eq.~(\ref{eq:Jabdefn}) that the on-site Casimir operators are just c-numbers $C_i=\frac{1}{4}\binom{2k}{2}$.
Thus the two-site effective Hamiltonian can be written as:
\begin{equation}
	P_{ij}^{(k)}=4k^2-4C_{ij}.
\end{equation}
Hence $P^{(k)}$ is sometimes referred to as the $\SO(2k)$ ferromagnetic Heisenberg model~\cite{fava2023nonlinear, swann2025}, since it is a nearest-neighbor Hamiltonian with terms proportional to the Casimir operator of the $\SO(2k)$ group, in analogy the well-known $\SU(2)$ Heisenberg model.

\subsection{Ground State Manifold}
The $k$-commutant for the free fermion systems under consideration is therefore obtained by maximizing the value of $C_{ij}=C_{ij}^{\mrm{max}}=k^2$ for each bond $\langle ij\rangle$.
The analysis in App.~\ref{app:repr}, which obtains the states that maximize the total Casimir on any $L$ sites, also shows that the value of $C_{ij}$ is maximized by states which are symmetric and have the same chiralities, i.e., $\bar{\Gamma}_i = \bar{\Gamma}_j$.
Through Lemma~\ref{lem}, this implies that the global ground states are also symmetric, and that each $C_{\tilde\imath\tilde\jmath}$ is maximized for any pair of sites $\tilde\imath\tilde\jmath$.
The ground states therefore also maximize $C_\mrm{tot}$, which indeed can be decomposed as a sum of all two-site Casimir operators, which are themselves maximized:
\begin{equation}
	C_\mrm{tot}=\sum_{i< j}^{2L}C_{ij}-L(L-1)k(2k-1),
\end{equation}
Maximizing the value of $C_\mrm{tot}$ selects two irreducible representations $V_{2L}^{2k,\pm}$ in the total Hilbert space, which, as we show in App.~\ref{app:repr}, have the ferromagnetic form
\begin{equation}
    V_{2L}^{2k, \pm} = \text{span}\{\ket{v}^{\ot 2L}:\ket v\in M_\mrm{MG}^{(k,\pm)}\},
\end{equation}
where states in $M_\mrm{MG}^{(k,\pm)}$ have chirality $\bar\Gamma_i = \pm 1$.
Hence the effective Hamiltonian $P^{(k)}$ is ferromagnetic.
The two subspaces $V_{2L}^{2k,\pm}$ can be mapped onto each other using the unitary parity operator ${\bar\Gamma^a=(-i)^L\prod_{i=1}^{2L}\bar\gamma_i^a=\tilde\Gamma^a}$, since it commutes with each $P_{ij}^{(k)}$ and anticommutes with each chirality operator $\bar\Gamma_i$.
This operation extends the $\SO(2k)$ replica symmetry to $\O(2k)$, under which the commutant,
\begin{equation}
\com{\mrm{MG}}\cong V_{2L}^{2k,+}\oplus V_{2L}^{2k,-},   
\end{equation}
forms a single irrep.
The dimensions of the two subspaces are therefore identical, and can be computed through the Weyl formula (cf.~Eq.~\eqref{eq:weyl-dim-formula})
\begin{equation}
	\dim(V_{2L}^{2k,\pm}) = \prod_{1\leq m< n\leq k}\frac{2L+2k-m-n}{2k-m-n}\label{eq:maj-dim}
\end{equation}
where $\dim(\com{\mrm{MG}})=\dim(V_{2L}^{2k,+})+\dim(V_{2L}^{2k,-})$.
Similar to the irreps $V_{2L}^{2k,\pm}$ which compose the ground state space, the ground state manifold splits into two disconnected components of opposite local chiralities
$M^{(k)}_\mrm{MG}=M^{(k,+)}_\mrm{MG} \sqcup M^{(k,-)}_\mrm{MG}$.
Using the fact that $\ker(P_{ij}^{(k)})=\ker(\mc L_{ij}^{(k)})$, we obtain using Eq.~(\ref{eq:Jabdefn}) that
\begin{equation}
	\ket{v}\in\mc H_\mrm{loc}:\ket{v}\in M^{(k)}_\mrm{MG} \iff \Lambda_{2k}(\ket{v}\ot\ket{v})=0
\label{eq:Lambda2kcondition}
\end{equation}
where $\Lambda_{2k}\defeq\sum_{a=1}^{2k}\bar\gamma^a\ot\bar\gamma^a$.
One might recognize this as a well-known condition that completely characterizes fermionic Gaussian states~\cite{lambda-op}.
We therefore find that $M_\mrm{MG}^{(k,+)}$ ($M_\mrm{MG}^{(k,-)}$) \textit{is the manifold of all fermionic Gaussian states of even (odd) parity on $2k$ Majorana modes}.
This characterization exposes a duality between real space and replica space in fermionic Gaussian systems that we will discuss in more detail in Sec.~\ref{sec:finaltrick}.
The generators $J^{ab}_{i}=-\frac{i}{2}\bar\gamma_i^a\bar\gamma_i^b$ of the $\so(2k)$ replica symmetry, are nothing more than quadratic Hamiltonians terms which couple the $2k$ copies of the original Majorana mode $\gamma_i$.
The associated group $\SO(2k)$ is therefore be the group of fermionic Gaussian unitaries on the given site acting across $2k$ replicas, and it acts transitively on the set of Gaussian states of a positive (negative) parity, which form the manifold $M_\mrm{MG}^{(k,+)}$ ($M_\mrm{MG}^{(k,-)}$).
Since the vacuum state is stabilized by the complete set of number-conserving Gaussian unitaries, which forms a $\U(k)$ subgroup~\cite{free-fermion-group}, $M_\mrm{MG}^{(k,+)}$ is the homogeneous manifold~\cite{ORTH-GRASSMANNIAN}
\begin{equation}
	M_\mrm{MG}^{(k,+)}\ \cong\ \SO(2k)/\U(k).
\label{eq:MFFgrass}
\end{equation}
Since the operator $\bar\Gamma^a$ maps product states into product states, the extended $\O(2k)$ symmetry can also be seen as acting on the ground state manifold.
Due to the addition of a discrete parity flipping operator, the action is transitive on the whole $M_\mrm{MG}^{(k)}$, resulting in
\begin{equation}
	M^{(k)}_\mrm{MG}\cong \OG(k,2k)\defeq \O(2k)/\U(k),\label{eq:gsman-maj}
\end{equation}
where $\OG(k,2k)$ is the \textit{orthogonal Grassmannian}, a manifold of (real) dimension $k(k-1)$.
An even more direct connection between this manifold and the set of fermionic Gaussian states on $2k$ Majorana models can be established, and is discussed in App.~\ref{app:grass}.
We can explicitly verify this for small values of $k$.
For $k = 1$, the connected component $\OG^+(1,2)$ is equal to a single point, corresponding to the operator $\1$ in the commutant $\mrm{Com}^{(1)}_{\mc U_{\rm MG}}$, and the point in $\OG^-(1,2)$ corresponds to the parity operator $(-i)^L\prod_{i=1}^{2L}\gamma_i$.
For ${k=2}$ it is the complex projective line ${\mathbb{CP}}^1$, i.e., the sphere $\mb S^2$, associated to the $\SU(2)$ symmetry used in Ref.~\cite{swann2025}, and for ${k=3}$ it is the complex projective space ${\mathbb{CP}}^3$.

\subsection{Extension to Parity-Flipping Gates}
So far we have studied the $k$-commutant of the Matchgate group $\mc U_{\mrm{MG}}\cong \SO(2L)$ generated by quadratic fermion Hamiltonians.
This group can be extended to the group $\mc U_{\rm MG^*}\cong \O(2L)$ by adding a discrete fermionic gate which is parity odd, such as the operator $Q=\gamma_i$~\cite{free-fermion-group}.
When replicated under our conventions, this becomes:
\begin{equation}
	(Q\ot Q^*)^{\ot k} = \bar\gamma^1_i\bar\Gamma^1\cdot \bar\gamma^2_i\cdot \bar\gamma^3_i\bar\Gamma^3\cdot ... \cdot \bar\gamma^{2k}_i.
\end{equation}
The addition of this gate restricts the $k$-commutant to be (cf. Eq.~\eqref{eq:unitary-comm})
\begin{equation}
    \!\!\mrm{Com}^{(k)}_{{\rm MG}^*} \!\!=\! \{\ket{X} \in \mrm{Com}^{(k)}_{\rm MG}\!:\!{(Q\ot Q^*)^{\ot k}\ket X\!=\!\ket X}\}.\!\!\label{eq:mgstar}
\end{equation}
Notice that ${((Q\ot Q^*)^{\ot k})^2=\1}$, so its eigenvalues are $\pm 1$.
To understand this space, we define the unitary operator 
\begin{equation}
W_k\defn e^{i\frac{\pi}{4}\bar\Gamma^1}e^{i\frac{\pi}{4}\bar\Gamma^3}\!...\,e^{i\frac{\pi}{4}\bar\Gamma^{2k-1}},
\end{equation}
which satisfies	$W_k\+(Q\ot Q^*)^{\ot k}W_k = \bar\Gamma_i$, and preserves the space $V_{2L}^{2k,+}\oplus V_{2L}^{2k,-}$, since each $\bar\Gamma^a$ does.
Then we have that
\begin{equation}
\ket X=W\ket Y \in \mrm{Com}^{(k)}_{{\rm MG}^*} \iff \ket Y\in V_{2L}^{2k,+}.
\end{equation}
In other words, for extended matchgates, we get
\begin{equation}
	\com{{\rm MG}^*} \cong W\cdot (V_{2L}^{2k,+}),
\end{equation}
so its dimension is given by Eq.~\eqref{eq:maj-dim}, and with this mapping in mind, we can associate the manifold {of Eq.~(\ref{eq:MFFgrass})} to it, hence we have
\begin{equation}
    {M^{(k)}_{{\rm MG}^*} \cong M^{(k,+)}_{\rm MG} \cong \SO(2k)/\U(k)}.
\end{equation}
Note that the need to introduce the $W_k$ operator is an artifact of the conventions used, since the $\SO(2k)$ replica symmetry only emerges naturally in terms of the \textit{anti-commuting} $\bar\gamma_i^a$ modes; see Appendix~\ref{app:conventions} for a more general discussion on the conventions.

\section{$\U(1)$-Symmetric Free Fermions} 
\label{sec:u1ff}

When charge conservation is enforced, the relevant fundamental modes are spinless fermions $\{c_i, c^\dagger_i\}_{i=1}^L$ which satisfy $\{c_i,c_j\}=0$ and $\{c_i,c_j\+\}=\delta_{ij}\1$.
The Hamiltonian terms are restricted to on-site and hopping terms of the form
\begin{equation}
	\begin{gathered}
    h_i^{(s)} \defn c^\dagger_i c_i - \frac{1}{2},\\
	h_{ij}^{(r)} \defn c^\dagger_i c_j + \text{h.c.},\quad
	h_{ij}^{(c)} \defn i c^\dagger_i c_j + \text{h.c.}
	\end{gathered}
\label{eq:freefermterms}
\end{equation}
These generate the group ${\mc U_\mrm{NC}\cong U(L)}$~\cite{free-fermion-group}.
Notice that since
\begin{equation}
[h_i^{(s)},h_{ij}^{(r)}]=-ih_{ij}^{(c)},\quad [h_i^{(s)},h_{ij}^{(c)}]=ih_{ij}^{(r)}\label{eq:asd-commutator}
\end{equation}
the sets of generators $\{h_i^{(s)},h_{ij}^{(r)}\}_{\langle ij\rangle}$ and $\{h_i^{(s)},h_{ij}^{(c)}\}_{\langle ij\rangle}$ produce the same group as $\{h_i^{(s)},h_{ij}^{(r)},h_{ij}^{(c)}\}_{\langle ij\rangle}$, and for some purposes it is sufficient to work with the smaller set of generators.
\subsection{Effective Hamiltonian}
In the replicated system composed of $2k$ copies of the Hilbert space, we can define the replica modes as follows
\begin{equation}
	\tilde c_i^a=\begin{cases}
		c_i^a \!\!\!\!&\text{if }a\text{ is odd,}\\
		(c_i^a)^T \!\!\!\!&\text{if }a\text{ is even.}
	\end{cases}
\end{equation}
This definition can be derived by transforming $\gamma\mapsto\tilde\gamma$ (cf.~Eq.~\eqref{eq:majorana-ad}) in the decomposition of fermionic operators into pairs of Majorana modes $c_i=\frac{\gamma_{2i-1}+i\gamma_{2i}}{2}$.
The adjoint operators corresponding to the terms in Eq.~\eqref{eq:freefermterms} then take the form\footnote{Note that the $-1/2$ contribution in $\mc L_{s,i}^{(k)}$ does not come from its counterpart in the on-site interaction, but rather from the canonical commutation relations applied to the even copies.}
\begin{equation}
\begin{aligned}
	\mc L_{s,i}^{(k)}&=\sum_{a=1}^{2k}\left(\tilde c_i^{a\dagger}\tilde c_i^a-\frac{1}{2}\right)\!,\\
	\mc L_{r,ij}^{(k)}&=\sum_{a=1}^{2k}(\tilde c_i^{a\dagger}\tilde c_j^a+\mrm{h.c.}),\\
	\mc L_{c,ij}^{(k)}&=\sum_{a=1}^{2k}(i\tilde c_i^{a\dagger}\tilde c_j^a+\mrm{h.c.}).
\end{aligned}\label{eq:u1-ls}
\end{equation}
As we did in Eq.~\eqref{eq:maj-kelin}, we can append the same Klein factors to also switch to $\bar c_i^a$ operators which anti-commute on different replicas, without changing the form of the adjoint operators.
To find $\com{{\rm NC}}$ as expressed in Eq.~\eqref{eq:comm-ad}, we start from the on-site adjoint operator $\mc L_{s,i}^{(k)}$.
Its kernel consists of states where the $2k$ replicas of site $i$ are restricted to the half-filling sector, i.e.,
\begin{equation}
    n_i\defeq \sum_{a=1}^{2k}\bar c_i^{a\dagger}\bar c_i^a=k.
\label{eq:fillingval}
\end{equation}
This condition effectively reduces the local Hilbert space dimension from $\mrm{dim}(\mc H_1^{\ot 2k})=2^{2k}$ to $\mrm{dim}(\mc H_\mrm{loc})=\binom{2k}{k}$. 
For the remaining adjoint operators, we get the following effective Hamiltonian on the reduced Hilbert space~\cite{fava2024}:
\begin{equation}
\begin{gathered}
	P_{ij}^{(k)} =\!\!\! \sum_{\alpha\in\{r,c\}}\!\!\big(\mc L_{\alpha,ij}^{(k)}\big)^2= 2k-4\sum_{a,b=1}^{2k} S^{ab}_i S^{ba}_j \\
    S^{ab}_i\defn\bar c_i^{a\dagger}\bar c_i^b-\frac{1}{2}\delta_{ab}
\end{gathered}	
\end{equation}
where $\{S^{ab}_i\}$ are local $\su(2k)$ generators, which commute on different sites and satisfy the relations
\begin{equation}
	[S_i^{ab}, S_i^{cd}] = \delta_{bc} S_i^{ad} - \delta_{ad} S_i^{cb}.\label{eq:gen-su2k}
\end{equation}
In terms of the $\su(2k)$ algebra, the $\binom{2k}{k}$-dimensional local Hilbert space is an irreducible representation, isomorphic to the $k$-th antisymmetric tensor representation ${\Lambda^k(\mathbb{C}^{2k})}$.
As in the previous section, to study the ground state space of $P^{(k)}$ we introduce the $\su(2k)$ Casimir operator
\begin{equation}
	C \defeq \sum_{a,b=1}^{2k} S^{ab}S^{ba}.
\end{equation}
As in the Majorana case, this Casimir operator is defined for any $\su(2k)$ representation: for on-site we define $C_i$ using the local generators $S^{ab} = S_i^{ab}$, for a pair of sites we get $C_{ij}$ using $S^{ab} = S^{ab}_i + S^{ab}_j$, and globally we get $C_{\tot}$ using $S^{ab} = \sum_i{S^{ab}_i}$.
Since the on-site representation is irreducible, the local Casimir operator  is just a number, which can be explicitly computed to be $C_{i}=k^2+\frac{k}{2}$, so that
\begin{equation}
	P_{ij}^{(k)}=4k(k+1)-2C_{ij}.
\end{equation}
For similar reasons as the Majorana case, $P^{(k)}$ can be seen as an $\SU(2k)$ ferromagnetic Heisenberg model.
\subsection{Ground state manifold}
The $k$-commutant for $\U(1)$-symmetric free fermions is therefore obtained by maximizing the two-site Casimir $C_{ij}=C_{ij}^\mrm{max}={2k(k+1)}$ for each bond $\langle ij\rangle$, where this maximization has been illustrated in App.~\ref{app:repr} for the Casimir on any $L$ sites.
As in the Majorana case, we find that the ground states of the two-site nearest-neighbor term are symmetric. By the application of Lemma~\ref{lem}, this implies the maximization of the Casimir on any pair of sites, which in turn also leads to the maximization of the global Casimir:
\begin{equation}
	C_\mrm{tot}=\sum_{i< j}^{L} C_{ij} - \frac{1}{2}L(L-2)k(2k+1).
\end{equation}
As we discuss in App.~\ref{app:repr}, we find that only a single irreducible representation $V_{L}^{2k}$ in $\mc H_\mrm{loc}^{\ot L}$ maximizes $C_\mrm{tot}$.
This representation is spanned by fully polarized vectors of the form $\ket{v}^{\ot L}$, so $P^{(k)}$ is ferromagnetic, and its ground state space is fully captured by the ground state manifold picture.
The dimension of $\com{\rm NC}\cong V_{L}^{2k}$ is computed through the Weyl formula, and is given by (cf.~Eq.~\eqref{eq:weyl-dim-formula}):
\begin{equation}
\dim(V_{L}^{2k})=\prod_{a=1}^{k}\prod_{b=k+1}^{2k}\frac{L+b-a}{b-a}.\label{eq:dim-u1}
\end{equation}
To find the corresponding ground state manifold $M^{(k)}_\mrm{NC}$, let us write $\mc L_{c,ij}^{(k)}$ in terms of the replica Majorana modes (recall $\bar c_i^a=\frac{\bar \gamma_{2i-1}^a+i\bar \gamma_{2i}^a}{2}$):
\begin{equation}
	\mc L_{c,ij}^{(k)} = \frac{i}{2}\sum_{x=1}^{4k} \bar\gamma_{i}^x\bar\gamma_{j}^x,\ \ \bar\gamma_i^x=\begin{cases}
		\bar\gamma_{2i-1}^x \!\!\!\!&\text{if }x\leq 2k,\\
		\bar\gamma_{2i}^{x-2k} \!\!\!\!&\text{if }x >2k,
	\end{cases}
\end{equation}
where the index $x$ bundles together all modes associated to the same site.
From the condition that $\com{\mrm{NC}}\subseteq\mrm{ker}(\mc L_{c,ij}^{(k)})$, we find again that
\begin{equation}
	\ket v\in\mc H_\mrm{loc}:\ket v\in M^{(k)}_\mrm{NC} \iff \Lambda_{4k}(\ket v\ot\ket v)=0
\label{eq:Lambda4kcondition}
\end{equation}
with the definition $\Lambda_{4k}\defeq\sum_{x=1}^{4k}\bar\gamma^x\ot\bar\gamma^x$.
Therefore, from the Gaussianity condition~\cite{lambda-op} and the condition of Eq.~(\ref{eq:fillingval}), the ground state manifold $M^{(k)}_\mrm{NC}$ \textit{is the manifold of all fermionic Gaussian states on $2k$ sites at half-filling}.
As discussed at the beginning of this section, since we applied the constraints derived from $\mc L_{s,i}^{(k)}$ and $\mc L_{c,ij}^{(k)}$, there is no need to also study $\mc L_{r,ij}^{(k)}$ (cf.~Eq.~\eqref{eq:asd-commutator}).
The $\su(2k)$ replica symmetry is the algebra of generators for quadratic number-preserving Hamiltonians on $\mc H_\mrm{loc}$, and the associated group $\U(2k)$ acts transitively on $M^{(k)}_\mrm{NC}$ (for simplicity we added a phase degree of freedom to the group, which corresponds to adding a trivial identity $\mathds{1}$ operator to the Lie algebra).
Since any half-filled state is stabilized by separately acting with number-conserving Gaussian unitaries on the $k$ empty and $k$ filled modes, which forms the group $\U(k)\times \U(k)$, we find the manifold to be~\cite{GRASSMANNIAN}
\begin{equation}
	M^{(k)}_\mrm{NC}=\Gr(k,2k)=\U(2k)/(\U(k)\times \U(k)),
\end{equation}
i.e., $M^{(k)}_\mrm{NC}$ is a \textit{complex Grassmannian} of (real) dimension~$2k^2$.
For $k=1$ for example, the ground state manifold is exactly the Bloch sphere associated to the $\U(1)$ commutant of Eq.~\eqref{eq:U1comm}.
A direct connection can also be established to the Grassmannian manifold can be established as follows.
Mathematically, $\Gr(k,2k)$ is the manifold that parametrizes all linear subspaces (hyperplanes) $W\subseteq \mb C^{2k}$ with $\mrm{dim}(W)=k$.
Given such a hyperplane $W$ with basis $\{e_1,...,e_k\}$, where $\{e_j\}$ are $2k$-dimensional complex vectors, $W$ is characterized by the antisymmetric exterior product (or ``wedge product'') of all its basis elements $e_1\land\cdots\land e_k$, where $e_i \land e_j = -e_j \land e_i$.
We can identify each basis vector $e_j$ with an orbital with creation operator $\cd_{e_j}$ in the single-particle fermionic Hilbert space of $2k$ fermions; hence the plane $W$ can be identified with a $k$-particle fermionic Gaussian state as
\begin{equation}
    e_1 \land \cdots \land e_k \sim c\+_{e_1}\dots c\+_{e_k}\ket{\mrm{vac}}.\label{eq:grassmannian-half-filling}
\end{equation}
This explicitly shows that the manifold of all Gaussian fermionic states with $k$ complex fermions in a system with $2k$ orbitals is the Grassmannian $\Gr(k,2k)$.
The relation of this manifold to the orthogonal Grassmannian from the Majorana case, is found Appendix~\ref{app:grass}, where we show how the fact that $\com{\mrm{MG}}\subset\com{\mrm{NC}}$ translates into a natural and physically meaningful embedding ${\OG(k,2k)\hookrightarrow\Gr(k,2k)}$.
\section{Free Fermions with Flavor Symmetries} 
\label{sec:flavorff}
It is possible to straightforwardly generalize the results of Sections~\ref{sec:mgff} and~\ref{sec:u1ff} to free fermion systems with larger internal symmetries.
Consider for example the generators of Eq.~\eqref{eq:u1-ls}: these can be seen as Hamiltonian terms describing interactions between fermions $\bar c^a_i$ with an $\SU(2k)$ flavor symmetry.
If we generalize the original system by putting $N$ fermion flavors per site $\{c^\alpha_i\}_{\alpha=1}^N$ with $\SU(N)$-symmetric terms of the form $\sum_{ij,\alpha} T_{ij}c_i^{\alpha\dagger}c_j^\alpha$ with all possible Hermitian $T$, then the adjoint operators take the form
\begin{equation}
	\mc L_{T, ij}^{(k)} = \sum_{a=1}^{2k}\sum_{\alpha=1}^N \left(T_{ij} \bar{c}_i^{\alpha,a\dagger}\bar c_j^{\alpha,a}+\mrm{h.c.}\right)\!,
\end{equation}
which is a fermionic interaction term with an $\SU(2kN)$ flavor symmetry.
The rest then follows in the same way as before, thus showing a correspondence between the $k$-commutant of fermions with all allowed $\SU(N)$ flavor symmetric terms, and the $kN$-commutant of spinless fermions with all allowed $\U(1)$ symmetric terms.
Similarly, the adjoint operators for the Majorana free fermions in Eq.~\eqref{eq:majorana-ad} can be interpreted as interactions between Majorana modes $\bar\gamma_i^a$ with an $\SO(2k)$ flavor symmetry.
The adjoint operators for free Majorana fermions with $N$ modes per site $\{\gamma^\alpha_i\}_{\alpha=1}^N$ and with $\SO(N)$-symmetric terms of the form $i\sum_{ij,\alpha}K_{ij}\gamma_i^{\alpha}\gamma_j^\alpha$ with all possible real antisymmetric $K$, are then
\begin{equation}
	\mc L^{(k)}_{K,ij} = i\sum_{a=1}^{2k}\sum_{\alpha=1}^{N} K_{ij}\bar\gamma_i^{\alpha,a}\bar\gamma_j^{\alpha,a}.
\end{equation}
This again can be interpreted as a free Majorana system with a larger $\SO(2kN)$ flavor symmetry, showing that the $k$-commutant of Majorana fermions with an $\SO(N)$ flavor symmetry is isomorphic to the $kN$-th commutant of regular Majorana fermions.
A summary of the $k$-commutants for all systems analyzed in this work can be found in Table~\ref{tab:manifolds}.

\section{Coherent State Resolution of the Commutant Projector} 
\label{sec:finaltrick}
\subsection{Commutant Projectors}
The $k$-commutant $\com{\mc U}$ is intimately connected to the $k$-th moments of the associated unitary group $\mc U$.
For any unitary group $\mc U$, the $k$-th moments are obtained using the \textit{twirling operation}, given by \cite{Mele2024introductiontohaar}
\begin{equation}
	\Phi^{(k)}_{\mc U} = \int_{\mc U} \dd U\, (U \ot U^*)^{\ot k},\label{eq:twirl}
\end{equation}
which can be interpreted as the $k$-th moment of a random unitary operator distributed according to the Haar measure on the group $\mc U$.
Note that for our purposes, it will be convenient to interpret $\Phi^{(k)}_{\mc U}$ as an operator acting on $2k$ copies of the Hilbert space.
Twirling operations naturally appear in the statistical analysis of quantum states, e.g., given a density matrix $\rho$, its twirl $\Phi^{(k)}_{\mc U}\ket{\rho}^{\ot k}$ in the vectorized Hilbert space encodes the $k$-th moment of the state's orbit under the action of the group $\mc U$. 
For instance, if $\mc U =\mc U_{\rm MG}$ is the group of fermionic Gaussian unitaries and $\ket{0}$ is a reference fermionic Gaussian state, then the projected state $\Phi^{(k)}_{\rm MG}\ket{0}^{\ot 2k}$ describes the $k$-th moment of the ensemble of random Gaussian states.
From the components of this state, one can compute the average over the ensemble of any observable which depends linearly on $\rho^{\ot k}$: examples include many-point correlators and measures for various quantum information resources, such as entanglement and non-stabilizerness \cite{Tarabunga2024criticalbehaviorsof,Haug2023stabilizerentropies,Chan_2024,Varikuti2024unravelingemergence}.
It is well-known~\cite{Collins_2006, HarrowLow2009, brandaoharrowhorodecki2016} that this twirling operators on $\mc U$ are exactly the orthogonal projectors onto the $k$-commutants as defined in Eq.~(\ref{eq:comm}):
\begin{equation}
    \Phi^{(k)}_{\mc U} \equiv \Pi^{(k)}_{\mc U},
\label{eq:twirlcomm}
\end{equation}
and we will use the two interchangeably.
Hence knowing the $k$-commutant generally allows for a simpler expression for this operation.
For example, when $\mc U$ is the complete set of Haar random unitaries on the Hilbert space $\mH_1^{\ot L}$, the commutant $\com{\rm Haar}$ is given by Eq.~(\ref{eq:Haarcommutant}), and the projector onto it reads \cite{Collins_2006,Mele2024introductiontohaar}
\begin{equation}
    \Pi^{(k)}_{\rm Haar} = \sum_{\sigma, \tau \in S_k}{\text{Wg}(\sigma^{-1} \tau)\ket{\sigma}^{\ot L} \bra{\tau}^{\ot L}},
\label{eq:Haarprojector}
\end{equation}
where the $\text{Wg}(\cdots)$ is the appropriate Weingarten function of the permutations, which appear essentially due to the fact that the elements $\ket\sigma^{\ot L}\in\com{\rm Haar}$ labeled by permutations in $S_k$ are not orthogonal. 
This leads to the powerful machinery of ``Weingarten calculus'' that enables numerous computations of various properties of random unitaries and states~\cite{ Collins_2010, PhysRevX.7.031016, hunterjones2019unitary, keyserlingk2018hydro, Braccia_2024,Zhou_2019}.
\subsection{Generalized Coherent States}
But having realized the commutant as the ground state of a ferromagnetic Hamiltonian gives us more options.
As discussed in the previous sections, the fact that $\com{\mc U}$ is spanned by fully-polarized vectors $\ket v^{\ot n}$ (where $n$ is the system size, which is $2L$ in the Majorana case or $L$ in the number-conserving case).
This set of states is invariant under the action of a Lie group $G$ [$\O(2k)$ or $\SU(2k)$ without or with particle number conservation], which also endows the  free-fermion $k$-commutants with the structure of a homogeneous symmetric manifold of the form
\begin{equation}
    M_{\rm FF}^{(k)}\cong G/H, 
\end{equation}
where FF indicates the matchgate or $\U(1)$-conserving free-fermion groups, and $H \subset G$ a Lie subgroup [$\U(k)$ or $\mrm S(\U(k) \times \U(k))$ respectively].
This structure means that the set of fully polarized ground states $\ket v^{\ot n}$ is a family of \textit{generalized coherent states}~\cite{generalized-coher-st}, which provides a projection formula through the associated resolution of identity:
\begin{equation}
\begin{split}
	\frac{1}{\mc D_k}\Pi_{\rm FF}^{(k)} = \int_{M^{(k)}_{\rm FF}} \dd v \,(\ketbra{v}{v})^{\ot n} = \qquad\qquad\qquad \\
	\qquad = \int_{G} \dd g\, \big(\,U_g\ket{v_0}\!\big)^{\ot n}\big(\!\bra{v_0}U_g\+\,\big)^{\ot n}\label{eq:coherent-state-projection}
\end{split}
\end{equation}
where $\mc D_k=\mrm{dim}(\com{\rm FF})$ and $\ket {v_0}\in M_{\rm FF}^{(k)}$ is a reference single-site coherent state.
While the expression for the commutant projector in Eq.~(\ref{eq:coherent-state-projection}) looks similar to the twirling operation in Eq.~(\ref{eq:twirl}), there are some important differences.
First, the domain of integration in Eq.~(\ref{eq:coherent-state-projection}) does not scale with system size $n$ as in Eq.~(\ref{eq:twirl}), but rather with the replica number $k$.
Further, the coherent states over which we integrate simply factorize along the spatial direction; therefore the complexity of expressions of the form $\bra{\Psi_l}\Pi_{\mc U}^{(k)}\ket{\Psi_r}$, which naturally appear in the evaluation of various physical quantities, remains constant in system size, as long as the boundary states $\ket{\Psi_{l,r}}$ are spatially local. 
For large system sizes, these types of integrals generally become amenable to saddle-point asymptotic approximations, due to the large exponent $n$.
However, the families of states over which Eq.~(\ref{eq:twirl}) and~(\ref{eq:coherent-state-projection}) integrate are in some sense dual: in the twirl we average $2k$ replicas of fermionic Gaussian states on $n$ modes, while in the coherent-state resolution we have a tensor product of $n$ identical fermionic Gaussian states on $2k$ modes.
An alternative projection formula onto the commutant, which more closely resembles Eq.~\eqref{eq:twirl} and better illustrates the duality between the replica and spatial directions, can be derived directly from the Schur character orthogonality \cite{hall2015lie,fulton2011representation,knapp2001lie}. For a compact Lie group $G$, the projector onto an irreducible representation is given by integrating the group action weighted by the character:
\begin{equation}
    \frac{1}{\mc D_k}\Pi_{\mc U}^{(k)} =\int_G \dd g \, \chi_{V_G}^*(g)\, U_g^{\ot n}
\end{equation}
where the character $\chi_{V_G}(g)$ is the trace of $U_g^{\ot n}$ restricted to the commutant, and can be computed directly from the highest weight associated to $V_G$ (cf.~Appendix~\ref{app:repr}).
\subsection{Application: Free Fermion Page Curves}
As an example, in Appendix~\ref{app:conti} we apply the coherent state projection formula Eq.~\eqref{eq:coherent-state-projection} to the computation of \textit{free fermion Page curves}~\cite{PhysRevB.103.L241118,PhysRevB.104.214306,PRXQuantum.3.030201,PhysRevResearch.5.013044} of the $k$-th Renyi entropies for large system size $L$.
We define the $k$-th purity as
\begin{gather}
	E_k(\ell) = \tr_A(\rho_A^k),\qquad \rho_A=\tr_{\bar A}(\rho),
\end{gather}
for a bipartition $\bar A=\{1,...,\ell\}$ and $A=\{\ell +1,...,L\}$, and compute the form of the function
\begin{equation}
    -\frac{1}{L}\frac{1}{k-1}\log(\overline{E_k(\ell)})=f_k(r)+ \mc O\left(\frac{\log L}{L}\right)
\end{equation}
where $r=\ell/L$ and the average is over fermionic Gaussian states, thus performing the \textit{annealed} average of the $k$-th Rényi entropy $S_k(\ell)\defn \frac{1}{1-k}\log(E_k(\ell))$ to leading order in $L$.
The quantity $\overline{E_k(\ell)}$ can be expressed as the overlap of the $k$-th moment of the Gaussian state ensemble $\Pi^{(k)}_\mrm{MG}\ket 0^{\ot 2k}$ with a domain-wall configuration $|\1_{\bar A}^e\1_{A}^\eta\rangle$ which describes the partial trace operation \cite{PhysRevX.7.031016,Zhou_2019, vardhan2021entanglement, vardhan_entanglement_2024}:
\begin{equation}
    \overline{E_k(\ell)} = \bra{\1_{\bar A}^e\1_{A}^\eta} \Pi^{(k)}_\mrm{MG}\ket 0^{\ot 2k}; 
\end{equation}
the precise forms of the expressions are in App.~\ref{app:conti}.
The saddle-point approximation applied to the resulting integral for large $L$ provides the solution Eq.~\eqref{eq:solution}, which implies
\begin{equation}
\begin{split}
    f_k(r)=-\frac{1}{k-1}\sum_{p>0}^{\frac{k-1}{2}} \big( r \log(\cos^2\theta_p^*)+\qquad\qquad\\
    \qquad+(1-r)\log(\sin^2(\theta_p^* - \omega_p/2)) \big)
\end{split}
\end{equation}
where $\theta_p^*$ solves the saddle-point equation Eq.~\eqref{eq:saddlepointeq} and $\omega_p$ is as in Eq.~\eqref{eq:omegapdef}.
We also compute $\overline{E_2(\ell)}$ exactly, thus replicating a result derived Ref.~\cite{lastres2026} using an orthonormal basis for the 2-commutant of matchgates.
\subsection{Alternate Parametrizations of Coherent States}
In the computations presented in App.~\ref{app:conti}, we perform the integral over the group $G\cong \O(2k)$ in Eq.~\eqref{eq:coherent-state-projection}, where the integration measure is simply the standard Haar measure.
However, the other representation in Eq.~\eqref{eq:coherent-state-projection}, i.e., the integral over the manifold $M_{\rm FF}^{(k)}$, can prove more useful in other contexts, so here for completeness we present an overview of this representation.
The integration over the manifold $M^{(k)}_{\rm FF}$ is done by first choosing a reference state $\ket{v_0}$, and then parametrizing the generalized coherent states using a $k\times k$ complex matrix $Z_{mn}$ \cite{alma991013887169706535,Berezin:1978sn,hua1963harmonic}. 
In the Matchgate case, where the coherent states are fermionic Gaussian states, $Z$ is an antisymmetric matrix ($Z^T = -Z$), and the coherent states take the form
\begin{equation}
	\ket{v} \propto \exp(\frac{1}{2}\sum_{m,n=1}^k Z_{mn} c^\dagger_m c^\dagger_n ) \ket{v_0}\label{eq:Zmatrix}
\end{equation}
for some choice of complex fermionic operators $\{c_m\}$ on the $2k$ Majorana modes.
This parametrization has $k(k-1)$ real degrees of freedom, which matches the dimension of the manifold $M_\mrm{MG}^{(k)}=\OG(k,2k)$.
Since this is a disconnected manifold, to reach all states one will actually need to choose two reference states $|v_0^{\pm}\rangle$ of opposite parity.
Similarly, in the $\U(1)$-conserving case, the relevant states in the $k$-commutant are half-filled Gaussian states on $2k$ modes, parametrized by a generic $k \times k$ complex matrix $Z$ yielding $2k^2$ degrees of freedom [same as the dimension of $M_\mrm{NC}^{(k)}=\Gr(k,2k)$].
If we choose $\ket{v_0}=c\+_1 \cdots c\+_k\ket{\mrm{vac}}$, then we have:
\begin{equation}
	\ket{v} \propto \exp(\sum_{m,n=1}^k Z_{mn} c^\dagger_{m+k} c_{n}) \ket{v_0}.
\end{equation}
With these parametrizations, the integration measures over these coherent states take the forms:
\begin{equation}
\begin{gathered}
	(\dd v)_{\rm MG}\propto\frac{1}{\det(\1_k+ZZ\+)^{k-1}}\prod_{1\leq m<n\leq k}\dd^2 Z_{mn},\\
	(\dd v)_{\rm NC}\propto\frac{1}{\det(\1_k+ZZ\+)^{2k}}\prod_{m,n=1}^k\dd^2 Z_{mn},
\end{gathered}
\end{equation}
where we omitted normalization factors, which are rather messy but however known in the literature \cite{alma991013887169706535,hua1963harmonic}.
\section{Discussion and Outlook} 
\label{sec:outlook}
In this work, we studied the $k$-commutants of multiple classes of free-fermion unitaries, such as matchgates, which are generated by quadratic Majorana terms with and without parity flipping terms, and particle number conserving free-fermion terms. 
We expressed the $k$-commutant as the ground state of an effective Hamiltonian acting on $2k$ copies of the system, and showed that the effective Hamiltonian is an $\SO(2k)$- or $\SU(2k)$-symmetric ferromagnetic model, similar to the ones that have appeared in the context of noisy quantum circuit dynamics~\cite{enriched-phases,swann2025,fava2024}.
To describe the $k$-commutant, we exactly solved for the ground states of this model, and showed that they are in direct correspondence with the manifold of fermionic Gaussian states on $2k$ sites, which are given by Grassmannians.
Moreover, we showed that the projector onto the $k$-commutant can be expressed as an integral over generalized coherent states parametrized by these manifolds, which allowed us to use the associated resolution of identity to express averages of physical quantities (such as the Rényi entanglement entropies) over the complete set of free-fermion states, without relying on random matrix theory or on orthonormal bases for the $k$-commutant.
This geometric structure makes the free-fermion commutants drastically different from many other well-known commutants in the literature such as the Haar or Clifford commutants, whose geometry is trivial (i.e., a discrete set of points), and which then necessitates careful orthonormalization of elements for construction of projectors onto the commutants, leading to techniques such as the Weingarten calculus~\cite{Collins_2006, Mele2024introductiontohaar}.
While these coherent state integrals can also be hard to evaluate exactly for arbitrary values of the replica number $k$, we expect two kinds of simplifications. 
First, the complexity of these integrals grows only with $k$ rather than system size $L$.
Second, they can be computed for large $L$ using saddle point methods similar to the ones demonstrated here, which is usually the regime of interest from a many-body perspective. 
There might also be efficient numerical methods to evaluate these integrals by sampling the coherent state efficiently, whose complexity should again only growth with $k$, and it would be interesting to explore such ideas in future work.
The relatively simple structure of the integral projection formula, also hints at several promising future directions, such as analyzing averages of other quantities such as entanglement negativity and non-stabilizerness, studying averages over the particle-conserving free-fermions ensemble, and possibly performing the replica limit~\cite{PhysRevX.7.031016, Zhou_2019} for the quenched averages for the Rényi or von Neumann entropies as $k\rt 1$.
On the physical front, this geometric understanding of the commutants as ground state manifolds should lead to  direct applications for the dynamics of noisy random circuits~\cite{lashkari2013towards, sunderhauf2019quantum, bauer2017stochastic, ogunnaike2023unifying, enriched-phases, swann2025,  moudgalya2024, fava2024, fava2023nonlinear}.
The effective Hamiltonians used in this work naturally appear in such settings, and the averaged late-time dynamics of various physical quantities are governed by low-energy spin-wave-like excitations on top of the appropriate $k$-commutants~\cite{ogunnaike2023unifying, moudgalya2024}.
Such effective Hamiltonians were analyzed using field theories in \cite{swann2025, fava2023nonlinear, fava2024, PhysRevX.15.021020, Zhang2022universal}, and in many cases with the addition of measurements and postselection on top of the unitary dynamics.
It would be interesting to check if any additional insights or simplifications can be obtained using our interpretation of their ground states in the unitary limit terms of fermionic Gaussian states, and how this structure changes with the addition of measurements.
This geometric manifold of the commutant should also be relevant for the study of the dynamical effect of interacting \textit{impurities} to the system, which should lead to interesting hydrodynamics features similar to such effects observed in the context of continuous symmetries~\cite{li2025dynamics}.  
On the mathematical front, it would also be interesting to explore the geometry of various other naturally occurring commutants. 
One class of interesting commutants might be subsets of matchgates, such as dropping the on-site terms in the $\U(1)$ cases, which appear in many natural physical settings. 
We know that such changes in the ensemble should enlarge the commutants~\cite{moudgalya2022from}, and it would be interesting to check if such cases also map onto manifolds like the ones that we have found here.
Another class of commutants would be the $k$-commutants for $k \geq 2$ under the addition of interactions to the unitaries. 
While their algebraic structure has been characterized in many cases, we find -- as we will discuss in an upcoming work \cite{lastres2026-2} -- that they still have interesting geometric structures leading to manifolds which are not homogeneous, but rather have singularities, a fact that has implications for the entanglement dynamics of such systems.
On a different note, the proof technique used here for the $k$-commutants of matchgate, by showing their ferromagnetic nature, might be generalizable to other $k$-commutants of interest. 
For example, we also observe that the Clifford commutants~\cite{bittel2025clifford} appears to have a ferromagnetic structure although with a discrete manifold, but it might be interesting in future work to check if they can be understood as ground states of some simple Hamiltonians, which might provide a simpler proof of the Clifford commutant.

\textit{Note added} -- After having obtained many of these results in a different context, we became aware of Ref.~\cite{poetri2026}, which shows that $\com{{\rm MG}^*}$ of Eq.~\eqref{eq:mgstar} is algebraically generated by an $\SO(k)$ group of operators.
While we were writing the manuscript, Ref.~\cite{larocca2026} appeared, which computes $\com{{\rm MG}}$ and $\com{{\rm NC}}$ in a similar framework as \cite{poetri2026}.
All our results are obtained independently with very different methods, and agree where there is overlap: in particular Eq.~\eqref{eq:maj-dim} can be shown to be equal to Eq.~(54) in \cite{poetri2026} and to be consistent with Eq.~(59) in \cite{larocca2026}.
Eq.~\eqref{eq:dim-u1} can be shown to be equal to Eq.~(32) in \cite{larocca2026}.

\begin{acknowledgments}
We particularly thank Poetri Tarabunga for bringing  this problem to our attention.
We are grateful to Lesik Motrunich and Poetri Tarabunga for useful discussions. We thank Frank Pollmann for collaboration on \cite{moudgalya2022}, and S.M. thanks Lesik Motrunich for collaboration on \cite{moudgalya2024} and Shreya Vardhan for collaboration on \cite{vardhan_entanglement_2024}.
We acknowledge support from the Munich Quantum Valley, which is supported by the Bavarian state government with funds from the Hightech Agenda Bayern Plus, and the Munich Center for Quantum Science and Technology (MCQST), supported by the Deutsche Forschungsgemeinschaft (DFG, German Research Foundation) under Germany’s Excellence Strategy--EXC--2111--390814868.
\end{acknowledgments}

\bibliography{refs-ffgs, newbib}

\onecolumngrid
\appendix

\section{Symmetry Lemma for the Ground state space} 
\label{app:proof}
We prove here Lemma~\ref{lem}, introduced in the main text.
While simple, it is very useful, since it allows one to prove a result on the structure of the global ground state by verifying a condition that can be checked locally on pairs of nearest-neighboring sites.
\begin{proof}
	[Proof of Lemma~\ref{lem}] Let $\Sigma_{ij}$ be the swap operator on sites $ij$, which acts on $\mc H_\mrm{loc}^{\ot L}$ as
	\begin{equation}
		\Sigma_{ij}(...\ket{r}_i\ket{s}_j...)=(...\ket{s}_i\ket{r}_j...).
	\end{equation}
	On two sites ($L=2$), $\mrm{Sym}^2(\mc H_\mrm{loc})$ -- the space of states symmetric under the exchange -- is exactly the $+1$ eigenspace of $\Sigma_{12}$.
    This means that $\mc K_{ij}=\bigcap_\alpha\ker(m_{\alpha,ij})$ is contained in the $+1$ eigenspace of $\Sigma_{ij}$.
    Therefore $\bigcap_{\alpha,\langle ij\rangle}\ker(m_{\alpha,ij}) = \bigcap_{\langle ij\rangle}\mc K_{ij}$ is contained in the $+1$ eigenspace of $\Sigma_{ij}$ for all adjacent $\langle ij\rangle$.
    If the adjacent sites forms a connected graph, then $\{\Sigma_{ij}\}_{\langle ij\rangle}$ generate the whole group of permutations $S_L$, and for any $\Sigma\in S_L$, its $+1$ eigenspace will fully contain $\bigcap_{\alpha,\langle ij\rangle}\ker(m_{\alpha,ij})$.
    Therefore $\bigcap_{\alpha,\langle ij\rangle}\ker(m_{\alpha,ij})$ is a subset of $\mrm{Sym}^L(\mc H_\mrm{loc})$.
    For a state $\ket\psi\in\mrm{Sym}^L(\mc H_\mrm{loc})$, belonging to $\bigcap_{\alpha,\langle ij\rangle}\ker(m_{\alpha,ij})$ means that:
	\begin{eqnarray}
		\forall\alpha,\langle ij\rangle: m_{\alpha,ij}\ket\psi =0,
	\end{eqnarray}
	but since for any $\Sigma\in S_L$ we have $\ket\psi=\Sigma\ket\psi$, we find the equivalent condition:
	\begin{eqnarray}
		0 = \Sigma\+ m_{\alpha,ij}\ket\psi =\Sigma\+ m_{\alpha,ij}\Sigma\ket\psi = m_{\alpha,\tilde\imath\tilde\jmath}\ket\psi,
	\end{eqnarray}
	where $\tilde\imath\tilde\jmath$ can be chosen arbitrarily, independently of $\langle ij\rangle$.
    This completes the proof.
\end{proof}

\section{Representation theory of Semisimple Lie Algebras and the Ground State Manifold} 
\label{app:repr}
In this appendix, we review standard elements of the representation theory of semisimple Lie algebras, which provide the mathematical foundation for the geometric properties of the $k$-commutants discussed in the main text.
For a comprehensive treatment, we refer the reader to textbooks on the subject \cite{fulton2011representation, knapp2001lie, hall2015lie}. A reader familiar with the subject may proceed directly to Section~\ref{app:reprcomm} for details on the $\so(2k)$ and $\su(2k)$ representations of interest.
\subsection{Review of basic concepts}
In the following, we will discuss the representation theory of a semisimple Lie algebra $\mf g$ with a Cartan subalgebra (i.e. a maximal Abelian Lie subalgebra) $\mf h \subset \mf g$.
The dimension of $\mf h$ is referred to as the \textit{rank} of $\mf g$, denoted by $r$.
Strictly speaking, the discussion below holds in general only when $\mf g$ is a \textit{complex} Lie algebra, whereas the examples we are interested in [$\so(2k)$ and $\su(2k)$] are real compact Lie algebras.
However, the same framework carries over to the cases we are interested in due to reasons we discuss below in Sec.~\ref{subsubsec:realcomplex}.
\subsubsection{Weights and Roots}
For any finite-dimensional representation $V$ of $\mf g$, the action of $\mf h$ is diagonalizable, and the simultaneous eigenspaces of $Y\in\mf h$ are called \textit{weight spaces}:
\begin{equation}
    V_\lambda = \{v \in V \mid Y v = \lambda(Y) v,\ \forall Y \in \mf h\},
\label{eq:wtdefn}
\end{equation}
where $\lambda \in \mf h^*$ is a linear map from $\mf h$ to $\mb C$, and is called a \textit{weight} of the representation.
Weights form a discrete set $w(V)$ in $\mf h^*$.
The non-zero weights of the adjoint representation $X\mapsto \mrm{ad}_X(\cdot)=[X,\cdot\,]\in\mrm{End}(\mf g)$ are called \textit{roots}, denoted by the set $\Phi$, and they label specific elements $X_\alpha \in \mf g$ satisfying $[Y, X_\alpha] = \alpha(Y) X_\alpha$ for all $Y \in \mf h$.
Note that while weights are dependent on the representation $V$, roots are intrinsic objects which only depend on the algebra $\mf g$, since they are weights of in a fixed (adjoint) representation.
Together with a basis $\{Y_m\}$ for $\mf h$, the set of all root generators $\{X_\alpha\}$ forms the Cartan-Weyl basis of $\mf g$ (up to normalization).
In the familiar context of $\su(2)$ (i.e., $\mf{sl}_{\mb C}(2)$), $\mf h$ is spanned by $S^z$, and the root generators are the raising and lowering operators $S^+$ and $S^-$, associated to a positive and a negative root in the one-dimensional $\mf h^*$.
More generally, for a Lie algebra of rank $r$, these roots can be represented as $r$-tuples of elements in $\mathbb{C}$ once a basis on $\mf h^\ast$ is assigned.
It is also known that for any finite-dimensional representation, all weights live in the real vector space $\mrm{span}_{\mb R}(\Phi)$, so that in the following we will not have to worry about the complex nature of $\mf h^*$.
By choosing a codimension 1 hyperplane in $\mf h^*$ that contains no roots, we can split the root system $\Phi$ into positive roots $\Phi^+$ and negative roots $\Phi^-$.
This induces a partial ordering on the weights, where $\lambda \succeq \mu$ if and only if $\lambda - \mu\in \mf h^*$ is a linear combination of positive roots with non-negative coefficients.
In any representation $V$, the generators associated to positive roots $\{X_\alpha\}_{\alpha \in \Phi^+}$ can be thought of as generalized ``raising operators'', which map vectors in the weight space $V_\lambda$ to $V_{\lambda+\alpha}$.
A \textit{highest weight} $\bar\lambda\in w(V)$ is defined as a weight that is not smaller than any other weight, and a \textit{highest weight vector} is any element $v_{\bar\lambda}\in V_{\bar\lambda}$, which is therefore  annihilated by all raising operators $X_\alpha v_{\bar\lambda} = 0$.
The dimension $\dim(V_{\bar\lambda})$ is the multiplicity $m_{\bar\lambda}$ of $\bar\lambda$, and if $V$ is irreducible, the highest weight is unique and its multiplicity is 1.
More generally, the theorem of highest weight implies that, given a finite-dimensional representation $V$ of $\mf g$, it decomposes into a direct sum of irreps $\{I(\bar\lambda)\}$ labeled by its highest weights (with multiplicity):
\begin{equation}
	V\cong\!\!\bigoplus_{\bar\lambda\in w_\mrm{hig.}(V)}\!\!m(\bar\lambda)I(\bar\lambda).
\end{equation}
In the case of $\su(2)$, the highest weight labeling a spin-$j$ irrep corresponds to the state $\ket{j,j}$ which maximizes $S^z$ and is annihilated by $S^+$.
Due to irreducibility, the subspace associated to a given $\bar\lambda$ is obtained by linear action of the algebra $\mf g$ on the highest weight vectors $v_{\bar\lambda}$:
\begin{equation}
	m(\bar\lambda)I(\bar\lambda)\cong \mrm{span}\{Xv_{\bar\lambda} : v_{\bar\lambda}\in V_{\bar\lambda},\ X\in\mf g\},\label{eq:spanall}
\end{equation}
just as all states in a spin-$j$ irrep of $\su(2)$ can be obtained by acting the Lie algebra on the highest weight state in it.
If the representation $V$ is a tensor product of smaller representations $V=\bigotimes_{i=1}^L V_i$, where $\mf g$ acts as $X=\sum_{i=1}^L X_i$, one can easily see that its weights are precisely given by all possible sums of weights in each factor:
\begin{equation}
	\forall\lambda\in w(V):\lambda = \sum_{i=1}^L (\lambda)_i,\quad (\lambda)_i\in w(V_i).
\label{eq:factordecomp}
\end{equation}
Again, note that on the other hand the roots do not change, since they are independent of the representation.
\subsubsection{Casimir elements and the Weyl dimension formula}
It is useful to also endow $\mf h^*$ with an inner product.
This is naturally derived from the Killing form, defined as $\kappa(X,Y) = \tr(\mrm{ad}_X \mrm{ad}_Y)$, where $\mrm{ad}_X$ is the adjoint representation of $X \in \mf g$, and for semisimple Lie algebras this is a non-degenerate bilinear inner product.
This form establishes a canonical isomorphism between $\mf h$ and its dual space $\mf h^*$: to every weight $\lambda \in \mf h^*$, we can associate a unique element $H_\lambda \in \mf h$ such that $\lambda(Y) = \kappa(H_\lambda, Y)$ for all $Y \in \mf h$.
This correspondence induces a natural inner product on the space of weights, defined by\footnote{For the compact real Lie algebras relevant to our study in the main text, this inner product is real and positive-definite, which means that $\mf h^*$ possesses a Euclidean product.} 
\begin{equation}
    ( \lambda, \mu ) \defeq \kappa(H_\lambda, H_\mu) = \lambda(H_\mu) = \mu(H_\lambda).
\label{eq:proddefn}
\end{equation}
With these notions, we can introduce the quadratic Casimir element $C$ (formally, an element of the universal enveloping algebra $U(\mf g)$), defined as $C = \sum_a X_a X^a$, where $\{X_a\}$ and $\{X^a\}$ are dual bases of $\mf g$ with respect to the Killing form.
In a Cartan-Weyl basis $\{Y_m,X_\alpha\}$, where $\{Y_m\}$ are chosen to be orthonormal and $\kappa(X_\alpha,X_{-\alpha})=1$, we can write:
\begin{equation}
	C=\sum_m Y_m^2+\sum_{\alpha\in\Phi} X_\alpha X_{-\alpha}.
\end{equation}
Since $C$ commutes with all elements of $\mf g$, Schur's lemma implies that its restriction to any irrep $I(\bar\lambda)$ is just a scalar multiple of identity $C(\bar\lambda) \1_{I(\bar\lambda)}$.
The value of the Casimir element on an irrep $I(\bar \lambda)$ is given by
\begin{equation}
    C(\bar\lambda) = ( \bar\lambda, \bar\lambda + 2\rho ),\;\;\;\rho \defn \frac{1}{2}\sum_{\alpha \in \Phi^+} \alpha
\label{eq:casimirexpression}
\end{equation}
where $\rho$ is known as the Weyl vector, and $\Phi^+$ are the positive roots.
Furthermore, the dimension of the irrep $I(\bar\lambda)$ is also computed from its highest weight using the Weyl dimension formula:
\begin{equation}
    \dim(I(\bar\lambda)) = \prod_{\alpha \in \Phi^+} \frac{( \bar\lambda + \rho, \alpha )}{( \rho, \alpha )}. \label{eq:weyl-dim-formula}
\end{equation}
The expressions in Eqs.~\eqref{eq:maj-dim} and \eqref{eq:dim-u1} in the main text are direct applications of this formula to the relevant $\mf{so}(2k)$ and $\mf{su}(2k)$ representations. In the next part we will discuss these two concrete examples in more detail.
\subsubsection{Real vs Complex Lie Algebras}\label{subsubsec:realcomplex}
So far we have reviewed the representation theory of complex semisimple Lie algebras.
However, the symmetry algebras relevant to this work, $\so(2k)$ and $\su(2k)$, are real compact Lie algebras.
This distinction does not invalidate our understanding, because the finite-dimensional complex representations of a compact real Lie algebra $\mf g$ are in one-to-one correspondence with the finite-dimensional complex representations of its complexification $\mf g\ot_{\mb R}\mb C$ [e.g., $\mf{so}(2k)$ complexifies to $\mf{so}_\mathbb{C}(2k)$, and $\mf{su}(2k)$ complexifies to $\mf{sl}_\mathbb{C}(2k)$].
For the unitary representations that we are interested in, the generators $Y_m$ of the Cartan-Weyl basis are Hermitian, while $X_\alpha\+=X_{-\alpha}$.
Further, due to the raising/lowering nature of $X_{\pm\alpha}$, one can show that $[X_\alpha,X_{-\alpha}]= H_\alpha$ (i.e. the dual element to $\alpha$ in $\mf h$).
We also have by definition of $X_\alpha$, that $[H_\alpha,X_\alpha]=\alpha(H_\alpha) X_\alpha =  (\alpha,\alpha)X_\alpha$ [see also Eq.~(\ref{eq:proddefn})].
Hence $\{X_\alpha,X_{-\alpha}, H_\alpha\}$ form an $\su(2)$ subalgebra. 
This allows us to show that the set of $X_{-\alpha}$ for $\alpha\in\Phi^+$ which kills a highest weight vector $v_{\bar\lambda}$ (which is anyway killed by $X_\alpha$ for $\alpha \in \Phi^+$) is given by the roots $\alpha$ orthogonal to $\bar\lambda$:
\begin{equation}
	w=X_{-\alpha}v_{\bar\lambda}\ \implies\ \norm{w}^2=v_{\bar\lambda}\+ X_{\alpha}X_{-\alpha}v_{\bar\lambda}=v_{\bar\lambda}\+ (X_{-\alpha}X_{\alpha}+cH_\alpha) v_{\bar\lambda}=\norm{v_{\bar\lambda}}^2 \bar\lambda(H_\alpha) \propto (\bar\lambda,\alpha).
\end{equation}
where we used the definitions in Eqs.~(\ref{eq:wtdefn}) and (\ref{eq:proddefn}), and that $X_\alpha v_{\bar\lambda} = 0$.
This result will be useful for constructing the ground state manifolds $M^{(k)}_{\mc U}$ geometrically, since it shows that the complex subalgebra which stabilizes the highest weight space $V_{\bar\lambda}$ is generated by $\{Y_m\}\cup\{X_\alpha:\alpha\in\Phi^+\}\cup\{X_{-\alpha}:(\bar\lambda,\alpha)=0,\,\alpha\in\Phi^+\}$.
But since Hermitian generators can only be constructed as a linear combination of both $X_{\alpha}$ and $X_{-\alpha}$, this means that the stabilizer subalgebra $\mf g_{\bar\lambda}\subseteq\mf g$ of the \textit{real} algebra, only complexifies to the algebra generated by $\{Y_m\}\cup\{X_{\pm\alpha}:(\bar\lambda,\alpha)=0,\,\alpha\in\Phi^+\}$.
By identifying the real stabilizer subalgebra in this way, we will then construct the orbit of a highest weight vector $v_{\bar\lambda}$ (which will correspond to $M^{(k)}_{\mc U}$) as the quotient of the Lie group associated to $\mf g$ by the stabilizer subgroup associated to $\mf g_{\bar\lambda}$.
\subsection{Commutant dimension and Ground state manifolds for free fermions} 
\label{app:reprcomm}

In both cases analyzed in the main text, we wish to maximize the Casimir element $C_\mrm{tot}$ of a tensor product representation where all factors are equal.
Since for a two-site system $C_\mrm{tot}=C_{ij}$, our analysis will also apply to the two-site Casimir, which we use in the main text.
Let $\mc H_\mrm{loc}$ be the local representation and $w(\mc H_\mrm{loc})$ the set of its weights, then using Eq.~(\ref{eq:casimirexpression}) for a highest weight $\bar\lambda = \sum_i{(\bar\lambda)_i}$ in the tensor representation we get:
\begin{equation}
	C_\mrm{tot}(\bar\lambda)=(\bar\lambda,\bar\lambda)+(\bar\lambda,2\rho)= \sum_{ij}((\bar\lambda)_i,(\bar\lambda)_j)+\sum_i((\bar\lambda)_i,2\rho).\label{eq:tomaximize}
\end{equation}
In the following we will see how this expression is maximized by choosing the same weight on each site, resulting in a fully polarized state.
\subsubsection{Matchgates/Majorana free fermions: $\mf{so}(2k)$}
%
With reference to the generators defined in Eq.~\eqref{eq:gen-so2k}, we can choose the Cartan subalgebra $\mf{h} \subset \so(2k)$ to be spanned by the $k$ mutually commuting generators as
\begin{equation}
    \mf h = \{Y_m \defn J^{2m-1,2m}\}_{m=1}^k,
\end{equation}
consistent with the fact that $\so(2k)$ has rank $k$.
If we introduce the dual basis vectors $\epsilon_m \in \mf{h}^*$ defined by $\epsilon_m(Y_n) = \delta_{mn}$, then the inner product on $\mf{h}^*$ defined in Eq.~\eqref{eq:proddefn} is simply $(\epsilon_m, \epsilon_n) = \delta_{mn}$, so elements of $\mf h^*$, in particular the weights and roots, can be written as $k$-tuples in $\mb R^k$.
From the commutation relations, the root system of $\so(2k)$ is $\alpha \in \{\pm(\epsilon_m + \epsilon_n),\pm(\epsilon_m - \epsilon_n)\}$ for $m\neq n$ (type $D_k$ in the Dynkin classification), and we define the positive roots to be
\begin{equation}
\Phi^+=\{\epsilon_m \pm \epsilon_n:1 \leq m < n \leq k\}.
\label{eq:Phi+so2k}
\end{equation}
Then we obtain that the Weyl vector defined in Eq.~(\ref{eq:casimirexpression}) in the basis ordered as $(\epsilon_1, \epsilon_2, \cdots, \epsilon_k)$ reads:
\begin{equation}
	\rho=\big(\,k-1,\,k-2,\,...,\,1,\,0\,\big).
\end{equation}
Since each $Y_m=-\frac{i}{2}\bar\gamma^{2m-1}\bar\gamma^{2m}$ can take values $\pm \frac{1}{2}$ on the two Majorana modes, the weights of the $2^k$-dimensional local Fock space $\mc H_\mrm{loc}$ have the form $\lambda = \frac{1}{2} \sum_{m=1}^k \sigma_m \epsilon_m$, where $\sigma_m \in \{\pm 1\}$.
In the local representation there are two highest weights, and these correspond to irreducible spinor representations with opposite parity $\bar\Gamma=\prod_{m=1}^k 2Y_m$.
The chosen set of positive roots in Eq.~(\ref{eq:Phi+so2k}), implies that the two highest weights are
\begin{equation}
	\bar\lambda^+=\left(+\frac{1}{2},...,+\frac{1}{2},+\frac{1}{2}\,\right)\quad\text{and}\quad \bar\lambda^-=\left(+\frac{1}{2},...,+\frac{1}{2},-\frac{1}{2}\,\right),
\label{eq:lambdabarexpr}
\end{equation}
since no other weights can be obtained by adding elements of $\Phi^+$ to either $\bar\lambda^{+}$ or $\bar\lambda^{-}$.
In the many-body Hilbert space $\mc H_\mrm{loc}^{\ot 2L}$, both terms in Eq.~\eqref{eq:tomaximize} can be maximized separately: the $(\bar\lambda,\bar\lambda)$ term is maximized when all the local components $\bar\lambda_i$ are equal, and the $(\bar\lambda,2\rho)$ term is maximized when the sign vector $\sigma_m$ of $\bar\lambda_i$ is $+1$ for all $m\in\{1,...,k-1\}$.
Therefore, the ground state space in $\mc H_\mrm{loc}^{\ot 2L}$ decomposes into two irreps of multiplicity 1 whose highest weights correspond to tensor products of identical vectors $(v_{\bar\lambda^\pm})^{\ot 2L}$, namely
\begin{equation}
	V_{2L}^{2k,+}:\ 2L\cdot\bar\lambda^+=\big(\!+\!L,\,...,\,+L,\,+L\,\big),\qquad
	V_{2L}^{2k,-}:\ 2L\cdot\bar\lambda^-=\big(\!+\!L,\,...,\,+L,\,-L\,\big).
\end{equation}
By applying the Weyl dimension formula (Eq.~\eqref{eq:weyl-dim-formula}) we immediately obtain the result of Eq.~\eqref{eq:maj-dim}. 
As explained in the main text (cf.~Eq.~(\ref{eq:Lambda2kcondition})), physically the local highest weight vectors $v_{\bar\lambda^\pm}$ are two specific fermionic Gaussian states of even/odd parity.
If we express the irreducibility condition Eq.~\eqref{eq:spanall} in terms of the group $\SO(2k)$
\begin{equation}
	V_{2L}^{2k,\pm} = \mrm{span}\{(U\ket{v_{\bar\lambda^\pm}})^{\ot 2L}:U\in \SO(2k)\}
\end{equation}
we find that each irrep $V_{2L}^{2k,\pm}$ is spanned by ``fully polarized'' states of the form $(U\ket{v_{\bar\lambda^\pm}})^{\ot 2L}$.
The set of states in $\mc H_\mrm{loc}$ traced out by the orbit of $\ket{v_{\bar\lambda^\pm}}$ is therefore exactly the ground state manifold $M^{(k, \pm)}_\mrm{MG}$, and together, they exactly form the set of fermionic Gaussian states (of both parities) on $2k$ modes.
Geometrically each orbit can be built as the homogeneous manifold $\SO(2k)/\U(k)$, since $\U(k)$ is the real subgroup generated by the stabilizer subalgebra for each $v_{\bar\lambda^\pm}$, where the stabilizer algebra can be determined by constructing roots that are orthogonal to $\bar{\lambda}^{\pm}$, as explained in Appendix~\ref{subsubsec:realcomplex}.
For example, for $v_{\bar\lambda^+}$ this subalgebra is spanned by all $Y_m$ and by all $X_\alpha$ where $\alpha\in\{\pm(\epsilon_m-\epsilon_n):1\leq m<n\leq k\}$, since these are the only roots orthogonal to $\bar\lambda^+$ of Eq.~(\ref{eq:lambdabarexpr}); this is exactly the root system associated to $\U(k)$.
The full ground state manifold can be described as in the main text [cf.~Eq.~\eqref{eq:gsman-maj}] by extending $\SO(2k)$ to $\O(2k)$, adding a unitary map from $V_{2L}^{2k,+}$ to $V_{2L}^{2k,-}$ (and vice versa), which renders the group action transitive on the whole ground state space, leading to the orthogonal Grassmannian manifold of Eq.~(\ref{eq:gsman-maj}).
\subsubsection{$U(1)$-symmetric free fermions: $\mf{su}(2k)$}  
With reference to the generators defined in Eq.~\eqref{eq:gen-su2k}, we can choose the Cartan subalgebra $\mf{h} \subset \su(2k)$ to be spanned by the mutually commuting generators 
\begin{equation}
    \mf h = \{Y_a \defn S^{aa}\}_{a=1}^{2k},
\end{equation}
which are not linearly independent, as they satisfy $\sum_{a=1}^{2k} Y_a=0$, consistent with the fact that $\su(2k)$ has rank $2k -1$.
If we introduce a dual set of vectors $\epsilon_a \in \mf{h}^*$ defined by $\epsilon_a(Y_b) = \delta_{ab}-\frac{1}{2k}$, then the inner product on $\mf{h}^*$ defined in Eq.~\eqref{eq:proddefn} is simply $(\epsilon_a, \epsilon_b) = \delta_{ab}-\frac{1}{2k}$ (these are again overcomplete since $\sum_{a=1}^{2k} \epsilon_a=0$, hence the $-\frac{1}{2k}$ term in the inner product for this standard basis).
The commutation relations show that the root system of $\su(2k)$ is $\alpha = \epsilon_a - \epsilon_b$ for $a\neq b$ (type $A_{2k-1}$ in the Dynkin classification), and we can define the positive roots to be 
\begin{equation}
\Phi^+=\{\epsilon_a - \epsilon_b:1 \leq a < b \leq 2k\}.
\label{eq:SU2kroots}
\end{equation}
The generators $S^{ab}$ are exactly the ones associated to the roots $\epsilon_a - \epsilon_b$.
Then the Weyl vector reads:
\begin{equation}
	\rho=\frac{1}{2}\sum_{a=1}^{2k}(2k-2a+1)\epsilon_a.
\end{equation}
Since each $Y_a=\bar c\+_a\bar c_a-\frac{1}{2}$ can take values $\pm \frac{1}{2}$ on each site, we can check that the weights of the $\binom{2k}{k}$-dimensional local Fock space $\mc H_\mrm{loc}$ have the form $\lambda = \sum_{a=1}^{2k} n_a \epsilon_a$, where $n_a \in \{0,1\}$ and $\sum_{a=1}^{2k} n_a=k$, which is necessary for the associated vector $v_\lambda$ to have the correct eigenvalue under $Y_a$:
\begin{equation}
	Y_a v_\lambda = \lambda(Y_a)v_\lambda = \sum_{b=1}^{2k} n_b\left(\delta_{ab}-\frac{1}{2k}\right)v_\lambda = \left(n_a -\frac{1}{2}\right)v_\lambda.
\end{equation}
In the local representation there is a unique highest weight, thus proving irreducibility. This highest weight is
\begin{equation}
	\bar\lambda=\sum_{a=1}^{k}\epsilon_a\ \implies\ n_a=(1,1,...,1,0,...,0,0).
\label{eq:barlambdasu2k}
\end{equation}
since no other weights can be obtained by adding elements of $\Phi^+$ to $\bar\lambda$ [cf.~Eq.~\eqref{eq:SU2kroots}].
As before, in the many-body Hilbert space $\mc H_\mrm{loc}^{\ot L}$, both terms in Eq.~\eqref{eq:tomaximize} are maximized separately by choosing the highest weight $L\cdot\bar\lambda$, associated to the fully polarized vector $(v_{\bar\lambda})^{\ot L}$, which spans the associated irrep $V_{L}^{2k}$.
By applying the Weyl dimension formula of Eq.~\eqref{eq:weyl-dim-formula} we immediately obtain the result of Eq.~\eqref{eq:dim-u1}.
As explained in the main text (cf.~Eq.~(\ref{eq:Lambda4kcondition})), physically the local highest weight vector $v_{\bar\lambda}$ is a specific fermionic Gaussian state at half-filling (the one where the modes $\{1,...,k\}$ are occupied).
If we express the irreducibility condition Eq.~\eqref{eq:spanall} in terms of the group $\SU(2k)$
\begin{equation}
	V_{L}^{2k} = \mrm{span}\{(U\ket{v_{\bar\lambda}})^{\ot L}:U\in \SU(2k)\}
\end{equation}
we find the ground state space $V_{L}^{2k}$ is spanned by fully polarized states of the form $(U\ket{v_{\bar\lambda}})^{\ot L}$.
The set of states in $\mc H_\mrm{loc}$ traced out by the orbit of $\ket{v_{\bar\lambda}}$ is therefore exactly the ground state manifold $M^{(k)}_\mrm{NC}$, and is exactly the set of fermionic Gaussian states on $2k$ sites with $k$ excitations.
Geometrically the orbit can be built as the homogeneous manifold
\begin{equation}
	\Gr(k,2k) \cong \SU(2k)/\S(\U(k)\times \U(k)) \cong \U(2k)/(\U(k) \times \U(k)),
\end{equation}
where $\S(\U(k)\times \U(k))$ is the real subgroup generated by the stabilizer subalgebra of $v_{\bar\lambda}$, determined by constructing roots that are orthogonal to $\bar{\lambda}$ of Eq.~(\ref{eq:barlambdasu2k}), as explained in Appendix~\ref{subsubsec:realcomplex}.
This subalgebra is spanned by all $Y_m$ and by all $X_\alpha$ where $\alpha=\pm(\epsilon_a-\epsilon_b)$ where either $1\leq a <b\leq k$ or $k< a <b\leq 2k$, since these are the only roots orthogonal to $\bar\lambda$.
Each of the two sets of roots produce a $\mf u(k)$ algebra, but the global tracelessness condition $\sum_{a=1}^{2k}\epsilon_a=0$ forces the total algebra to just be $\su(k)\oplus\su(k)\oplus\mf u(1)$, which corresponds to $\S(\U(k)\times \U(k))$.
\section{Relationship between the Grassmanian manifolds for Free-fermion Commutants} 
\label{app:grass}
In the main text (cf.~Eq~\eqref{eq:grassmannian-half-filling}) we have shown how the space of fermionic Gaussian states at half-filling on $2k$ sites (i.e., the ground state manifold $M^{(k)}_{\rm NC}$ for number conserving free fermions) can naturally be understood as the Grassmannian $\Gr(k,2k)$, which parametrizes all $k$-dimensional hyperplanes $W$ in $\mb C^{2k}$.
Here we will discuss how the orthogonal Grassmannian $\OG(k,2k)$ (i.e. the ground state manifold for matchgates/Majorana free fermions) can be seen as submanifold of $\Gr(k,2k)$ in a physically meaningful way.
As we have seen in the main text, this is the manifold $M^{(k)}_{\rm MG}$ of fermionic Gaussian states on $2k$ Majorana modes.
Given a non-degenerate bilinear inner product $B(\cdot,\cdot)$ on $\mb C^{2k}$, $\OG(k,2k)$ parametrizes the \textit{isotropic} subspaces of dimension $k$, where $\forall u,w\in W:B(u,w)=0$, so naturally $\OG(k,2k)\subseteq\Gr(k,2k)$.
To see this physically, let us split the complex fermionic operators into Majorana modes $\{c_a=\frac{\gamma_{a,1}+i\gamma_{a,2}}{2}\}_{a=1}^{2k}$.
Given a fermionic Gaussian state $\ket v_i$ on the Majorana modes $\{\gamma_{a,i}\}_{a=1}^{2k}$ (for a given $i\in\{1,2\}$), let $D_{na}$ be the $k\times 2k$ Bogoliubov transformation matrix to get $k$ complex fermionic operators such that:
\begin{equation}
	f_{n,i} = \sum_{a=1}^{2k} D_{na}\gamma_{a,i},\quad f_{n,i}\ket v_i=0\;\;\forall\ n\in\{1,...,k\}.
\end{equation}
Then we can use the same matrix $D$ to define the complex fermionic operators $b_n=\sum_{a=1}^{2k} D_{na} c_a$ and $d_n\+=\sum_{a=1}^{2k} D_{na} c_a\+$, where again $n\in\{1,...,k\}$, which satisfy
\begin{equation}
	b_n (\ket v\ot\ket v)= \frac{f_{n,1}+if_{n,2}}{2}(\ket v\ot\ket v)=0,\qquad d_n\+ (\ket v\ot\ket v)= \frac{f_{n,1}-if_{n,2}}{2}(\ket v\ot\ket v)=0.
\end{equation}
The replicated state $\ket v\ot\ket v$ is therefore proportional to the fermionic Gaussian state $d_{1}\+ \cdots d_n\+\ket{\mrm{vac}}$, where $\ket{\mrm{vac}}$ is the state on $2k$ complex fermionic modes which is annihilated by $\{b_n,d_n\}_{n=1}^k$. 
This is precisely a half-filled fermionic Gaussian state, which belongs to the manifold $M^{(k)}_{\rm NC}$ for number conserving free-fermions.
Thus we have the embedding $M^{(k)}_\mrm{MG}$ for the Majorana free fermions in the larger manifold $M^{(k)}_\mrm{NC}$ for the $\U(1)$-conserving free fermions.
It is easy to verify that the canonical commutation relations of the $\{f_{n,i}\}$ complex fermions impose $DD\+=\1_k$ and $DD^T=0$.
Hence the rows of the matrix $D$ are an orthonormal basis of a subspace $W\subseteq \mb C^{2k}$ of dimension $k$ (since $DD\+=\1_k$), and this subspace is isotropic under the bilinear product $B(v,w)=v^Tw$ (since $DD^T=0$).
Thus the embedding $\ket v\mapsto\ket v\ot\ket v$ corresponds exactly to the natural embedding $\OG(k,2k)\hookrightarrow \Gr(k,2k)$.

\section{Replica formalism for fermionic Hilbert spaces} 
\label{app:conventions}

In this work, we are working with Hilbert spaces which have the natural structure of fermionic Fock spaces.
This may introduce some ambiguity in the definition of the $k$-commutant, since there are two different but natural ways to define the replicated Hilbert space.
To address the heart of the matter, we consider a generic physical system $\mc H$ (i.e., one replica) where a global parity operator $\Gamma$ is defined.
We assume $\Gamma^2=1$, and define the parity of an operator $X$ based on its commutation with $\Gamma$:
\begin{equation}
	\text{Parity even:}\;\;[\Gamma,X]=0,\qquad \text{Parity odd:}\;\;\{\Gamma,X\}=0.
\end{equation}
Notice that the product of operators of the same parity is parity even, and the product of operators of different parities is parity odd.
We will write $|X|=0$ for parity even operators and $|X|=1$ for parity odd operators.
Similarly we can define parity even/odd states based on the associated $\pm1$ eigenvalue of $\Gamma$.
In order to define the $k$-commutant, we will have to replicate the physical system $k$ times; we can do so in two different ways:
\begin{equation}
	\mc H^{\ot k}=\mc H\ot\mc H\ot...\ot\mc H,\qquad \mc H^{\got k}=\mc H\got\mc H\got...\got\mc H.
\end{equation}
The first is the usual tensor product, while the second is a \textit{$\mb Z_2$-graded tensor product}~\cite{graded-tprod}; what this means in practice is that the product of replicated operators behaves differently in the two cases:
\begin{equation}
	(X\ot Y)(R\ot S) = (XR)\ot (YS),\qquad (X\got Y)(R\got S) = (-1)^{|Y||R|}(XR)\got (YS).
\end{equation}
As long as all operators under consideration are parity even, there is therefore no distinction between the algebras of operators replicated using the two kinds of tensor products.
To indicate the two kinds of operator $X$ acting on the replica $a$ we will use the notations:
\begin{equation}
	X^a=\1\ot...\ot X\ot ...\ot \1,\qquad \gg X^a=\1\got...\got X\got ...\got \1.
\label{eq:XXhat}
\end{equation}
We can see how these structures apply to a system of $2L$ Majorana modes $\{\gamma_i\}_{i=1}^{2L}$. The parity is $\Gamma=(-i)^L\prod_{i=1}^{2L}\gamma_i$, and Majorana-string operators are parity even (odd) if the length of the string is an even (odd) number.
Then for $a\neq b$, if we choose to replicate the system with the usual tensor product, we will have $[\gamma_i^a,\gamma_j^b]=0$ (as in Ref.~\cite{larocca2026} and in the main text), while with the graded tensor product we will find $\{\gg\gamma_i^a,\gg\gamma_j^b\}=0$ (as in Ref.~\cite{poetri2026}).
Through the introduction of Klein factors analogous to the ones of Eq.~\eqref{eq:maj-kelin}, it is possible to describe the graded operator algebra using regular tensor products. Given an operator $X$, we define
\begin{equation}
	\bar X^a=\begin{cases}
		\1\ot...\ot \1\ot X\ot \1\ot ...\ot \1,\quad&\text{if }|X|=0,\\
		\Gamma \ot...\ot \Gamma \ot X\ot \1\ot ...\ot \1,\quad&\text{if }|X|=1.\\
	\end{cases}
\end{equation}
This family of operators satisfies the same commutation relations as $\{\gg X^a\}$ defined in Eq.~(\ref{eq:XXhat}).
The mapping can be extended to any operator string (and hence by linearity to the whole algebra of operators) by multiplying operators on different replicas in a fixed order, e.g., starting from $a=1$ left to right.
The map $\gg X^a\mapsto \bar X^a$ thus becomes an isomorphism of the operator algebras on $\mc H^{\got k}$ and on $\mc H^{\ot k}$.
This is also the structure behind the well-known Jordan-Wigner transformation.
Let us now show that the $k$-commutants defined with either convention are isomorphic operator algebras, assuming that all operators $U\in \mc U$ in the ensemble of unitaries of interest have definite parity.
One can also easily verify that
\begin{equation}
	V\defn\exp(i\frac{\pi}{4}\Gamma) \;\;\implies\;\;V\+ X V=\begin{cases}
		X,\quad&\text{if }|X|=0,\\
		iX\Gamma,\quad&\text{if }|X|=1.\\
	\end{cases}
\end{equation}
By defining $W=W_1 W_2 \cdots W_k$ with $W_a=\1^a$ if $a\equiv k\,(\mrm{mod}\, 2)$ and $W_a=V^a$ otherwise, we can therefore implement the mapping between the two conventions as a unitary transformation $W\+ U^{\ot k}W\propto\bar U^1\bar U^2...\bar U^k$. And since:
\begin{equation}
	[U^{\ot k},X]=0\quad\iff\quad [W\+ U^{\ot k}W,W\+ XW]=0,
\end{equation}
we have that $X$ belongs to the $k$-commutant for the replicated group $\mc U^{\ot k}$ if and only if $W\+ XW$ belongs to the $k$-commutant for the ``graded'' replicated group $\mc U^{\got k}$.
In this work we chose to use the normal tensor product convention for two main reasons.
First, because it is the more natural choice when the fermionic operators arise from the Jordan-Wigner transformation of spin operators on a chain, since replicated spin chains do not possess a natural fermionic Fock space structure.
This point of view is also adopted in previous works on many-body dynamics \cite{enriched-phases,fava2023nonlinear, fava2024,swann2025}.
Second, it is natural to describe operators $X\in\mrm{End}(\mc H)$ as elements in $\mc H\ot\mc H^*$, and our vectorization approach to the problem, where we express operators $X$ as vectors $\ket X\in\mc H^{\ot 2}$ (cf. Eq.~\eqref{eq:vectorization}), would therefore clash with the choice of a graded tensor product for the replicas.
As discussed in Ref.~\cite{graded-tprod}, it is possible to endow $\mrm{End}(\mc H)$ with a natural graded tensor structure, but this leads to complications in the definition of the adjoint operators $\mc L^{(k)}_{\alpha,ij}$, which we wish to avoid: e.g., for $a\neq b$ the state $\ket{\gg\gamma^a}$ would not be killed by the adjoint action of $\gg\gamma^b$, which would then act non-locally across copies.
\section{Details on the commutant projection computations} 
\label{app:conti}
In this appendix we will compute the average over all fermionic Gaussian states $\rho=\ketbra{\psi}{\psi}$ on a system of $2L$ Majorana modes (hence $L$ physical sites) of the $k$-th purity, introduced in the main text:
\begin{gather}
	E_k(\ell) = \tr_A(\rho_A^k) = \braket{\1^e_{\bar A}\1^\eta_{A}}{\rho^{\ot k}},\qquad \rho_A=\tr_{\bar A}(\rho),
\end{gather}
where $\bar A=\{1,...,\ell\}$ and $A=\{\ell +1,...,L\}$, and $\ket{\1^e_{\bar A}\1^\eta_{A}}$ is a ``domain-wall'' state that implements the partial trace over $\bar A$ and the cyclic trace over $A$, which we define more precisely below.
The average of this quantity over all fermionic Gaussian states is obtained by expressing $\ket{\psi} = U\ket{0}$ for some reference state $\ket{0}$, and averaging over the unitaries $U$ from the group $\mc U_{\rm MG}$. 
Hence the average of $E_k(\ell)$ can be expressed as:
\begin{equation}\label{eq:tocompute}
	\overline{E_k(\ell)} = {\int_{\mc U_{\rm MG}}\mrm{d}U\ \bra{\1^e_{\bar A}\1^\eta_{A}}(U \ot U^\ast)^{\ot k}\ket{\rho^{\ot k}}} = \langle{\1^e_{\bar A}\1^\eta_{A}}\vert\, \Pi^{(k)}_{\rm MG}\ket{0}^{\otimes 2k},
\end{equation}
where $\Pi^{(k)}_{\rm MG}$ is the projector onto the $k$-commutant of $\mc U_{\rm MG}$, and we have used Eq.~(\ref{eq:twirlcomm}).
In the following, we will evaluate this quantity using generalized coherent states.
\subsection{Characterization of the boundary states} 
The state $\ket{\1^e_{\bar A}\1^\eta_{A}}$ is a domain wall configuration of the form $\ket{\1^e_{\bar A}\1^\eta_{A}}=\ket{\1^e_1} \dots \ket{\1^e_\ell} \vert\,\1^\eta_{\ell +1}\rangle\dots\ket{\1^\eta_L}$.
In terms of the replicated basis, these states correspond to different pairings of the $2k$ copies.
If we denote a generic local basis state for a single copy by $\ket{n^a}$, the identity pairings are given by~\cite{vardhan_entanglement_2024}:
\begin{equation}
	\ket {\1^e} = \sum_{\{n^a\}} \delta_{n^1n^2}\,...\,\delta_{n^{2k-1}n^{2k}} \ket{n^1}\ket{n^2}...\ket{n^{2k}},\qquad
	\ket {\1^\eta} = \sum_{\{n^a\}} \delta_{n^2n^3}\,...\,\delta_{n^{2k}n^{1}} \ket{n^1}\ket{n^2}...\ket{n^{2k}}.
\label{eq:eetapairing}
\end{equation}
These states can be thought of as appropriate tensor product of maximally entangled states of the form 
\begin{equation}
    \ket{I} = \sum_{n}{\ket{n}_u\ket{n}_d},
\end{equation}
which is just the identity operator represented as a state on a doubled Hilbert space of a single physical site, one corresponding to the `ket' space (labelled by $u$) and one to the `bra' space (labelled by $d$).
To describe these pairings in the fermionic language, let us first focus on this single state $\ket{I}$.
Algebraically, the identity operator trivially commutes with any local Majorana operator $\gamma_i$ on the two Majorana sites $i\in\{1,2\}$ on which $\ket I$ lies, such that $(\gamma_i I - I \gamma_i) = 0$.
Under the state-operator mapping, left-multiplication corresponds to acting on the `ket' with $\tilde\gamma_i^u$, while right-multiplication corresponds to acting on the `bra' with $\tilde\gamma_i^d$ [cf.~Eq.~\eqref{eq:majorana-ad}], thus
\begin{equation}
	(\tilde\gamma_i^u-\tilde\gamma_i^d)\ket I = 0\;\;\implies\;\; \tilde\gamma_i^u\ket I=\tilde\gamma_i^d\ket I.
\end{equation}
If we introduce the Klein factors of Eq.~\eqref{eq:maj-kelin}, this implies the local stabilizer conditions:
\begin{equation}
	\bar\gamma_1^u \bar\gamma_1^d \bar\gamma_2^u \bar\gamma_2^d \ket I =
    \tilde\gamma_1^u \tilde\Gamma^u\tilde\gamma_1^d \tilde\gamma_2^u \tilde\Gamma^u\tilde\gamma_2^d \ket I =
    -\tilde\gamma_1^u \tilde\gamma_1^d \tilde\gamma_2^u \tilde\gamma_2^d \ket I =
    -\ket I.
\end{equation}
We can now apply this stabilizer logic to our specific configurations on $2k$ copies.
For $\ket{\1^e}$, copy $2m-1$ is paired with copy $2m$, while $\ket{\1^\eta}$ cyclically connects copy $2m$ to $2m+1$ in Eq.~(\ref{eq:eetapairing}).
Given the above constraint on the physical sites composed of two Majorana sites, we have some freedom to isolate the Majorana structure of the replicas.
One way to do that is to define the states $\ket{e}$ and $\ket\eta$ of positive parity on a single replicated Majorana site such that they satisfy the quadratic stabilizer conditions~\cite{swann2025}:
\begin{equation}
\begin{gathered}
	\ket{\1^e_j}=\ket{e_{2j-1}}\ket{e_{2j}},\qquad \forall m\in\{1,...,k\}: -i\bar\gamma^{2m-1}\bar\gamma^{2m}\ket{e}=\ket{e},\\
	\ket{\1^\eta_j}=\ket{\eta_{2j-1}}\ket{\eta_{2j}},\qquad \forall m\in\{1,...,k-1\}: -i\bar\gamma^{2m}\bar\gamma^{2m+1}\ket{\eta}=\ket{\eta},\quad -i\bar\gamma^{1}\bar\gamma^{2k}\ket{\eta}=\ket{\eta}.\label{eq:estab}
\end{gathered}
\end{equation}
If we define the complex fermions at a given site $i$ by pairing replicated Majorana modes as:
\begin{equation}
	d_{m,i}=\frac{1}{2}(\bar\gamma_i^{2m-1}+i\bar\gamma_i^{2m}),\quad m\in\{1,...,k\},\label{eq:mymodes}
\end{equation}
then $\ket{e}$ is the vacuum state, with zero excited modes.
Both states $\ket{e}$ and $\ket\eta$ are fermionic Gaussian states on the $2k$ replica modes, so any of them can be chosen as reference states for the integral formula for the projector Eq.~\eqref{eq:coherent-state-projection}. We will choose the vacuum state $\ket{e}$, and write
\begin{equation}
	\Pi^{(k)}_{\rm MG}	= \mc D_k\int_{\mrm O(2k, \mb R)} \dd R\, \big(\,U_R\ket{e}\!\big)^{\ot 2L}\big(\!\bra{e}U_R\+\,\big)^{\ot 2L},\label{eq:ourprojectr}
\end{equation}
where $U_R$ is the unitary representation of $R$.
For the state $\ket 0$ in Eq.~(\ref{eq:tocompute}), we choose the fermionic Gaussian state $\ket 0=\ket{\phi_0}^{\ot L} \in \mH_L$ which satisfies $-i\gamma_{2j-1}\gamma_{2j}\ket 0=\ket 0$ for all $j\in\{1,...,L\}$. 
This state has no entanglement between different physical sites labeled by $j$.
In the replica space, $\ket 0^{\ot 2k}$ can be characterized again in terms of stabilizers as:
\begin{equation}
	\forall j\in\{1,...,L\}, a\in\{1,...,2k\}: \quad(-1)^ai\bar\gamma_{2j-1}^a\bar\gamma_{2j}^a\ket 0^{\ot 2k}=\ket 0^{\ot 2k}.\label{eq:0statestab}
\end{equation}
\subsection{Parametrization of coherent states} 
Using Eq.~\eqref{eq:ourprojectr}, we can express the averaged observables of Eq.~\eqref{eq:tocompute} as,
\begin{equation}
	\overline{E_k(\ell)} =  \mc D_k\int_{\O(2k, {\mb R})}\!\!\!\!\dd R\ \Omega_e^{2\ell}\,\Omega_\eta^{2(L-\ell)}\,\Omega_0^{L},
\label{eq:Elexpr}
\end{equation}
where $\dd R$ is the Haar measure over $\O(2k, \mb R)$, $\mc D_k$ is the dimension of the Matchgate commutant [cf.~Eq.~\eqref{eq:maj-dim}], and
\begin{equation}
	\Omega_{e}(R)\defeq \bra{e}U_R\ket{e},\quad \Omega_{\eta}(R)\defeq \bra{\eta}U_R\ket{e}, \quad \Omega_{0}(R)\defeq \bra{e}\bra{e}(U_R\+\ot U_R\+)\ket{\phi_0}^{\ot 2k}.\label{eq:definition-omega}
\end{equation}
Note that the complexity of its evaluation only scales with $k$, and the factorization in the spatial direction provides simplifications that we illustrate below.
Also note that since all states in Eq.~\eqref{eq:definition-omega} have even parity, the overlap is zero if $R\notin\SO(2k,\mb R)$.
As a matter of convention, we will choose the reference state $\ket e$ and the coherent state $\ket \eta$ to have norms $\braket{e}{e} = \braket{\eta}{\eta} = 1$.
However, since in the original expression one must have $\braket{\1^e}{\1^e}=\braket{\1^\eta}{\1^\eta}=\tr(I^{\ot k})=2^k$, we shift this normalization factor to $\ket{\phi_0}^{\ot 2k}$ instead, thus setting $\Vert {\ket{\phi_0}^{\ot 2k}}\Vert^2=2^k$ in $\Omega_0(R)$.
These matrix overlaps can be parametrized using standard free-fermion methods~\cite{Oi_2012,AAAA}.
First, we choose Eq.~\eqref{eq:mymodes} to be our reference set of complex fermionic operators, so that $\ket e$ is the vacuum.
For $R\in\O(2k, \mb R)$, we define the action of the unitary representation $U_R$ on Majorana modes as:
\begin{equation}
	\binom{\{U_R \bar\gamma^{2m-1}U_R\+\}}{\{U_R \bar\gamma^{2m}U_R\+\}} = R \binom{\{\bar\gamma^{2m-1}\}}{\{\bar\gamma^{2m}\}},
\end{equation}
where we chose to place Majorana modes with and odd index on the top half of the vector, and modes with an even index on the bottom half.
In the basis associated to the $d$-modes, the orthogonal matrix $R\in\O(2k, {\mb R})$ has the form:
\begin{equation}
	T\!\defeq\!\frac{1}{2}\begin{pmatrix} \1_k & i\1_k \\ \1_k & -i\1_k \end{pmatrix}\!,\ \ 
	\binom{\{d_m\}}{\{d_m\+\}} \!= T \binom{\{\bar\gamma^{2m-1}\}}{\{\bar\gamma^{2m}\}} \implies
	\binom{\{U_R d_m U_R\+\}}{\{U_R d_m\+ U_R\+\}} \!=\tilde R \binom{\{d_m\}}{\{d_m\+\}},\ \ 
	\tilde R\defeq TRT^{-1} \!=\! \begin{pmatrix} A & B \\ B^* & A^* \end{pmatrix}\!,\label{eq:tilderotation}
\end{equation}
where $A$ and $B$ are complex $k\times k$ matrices.
Note that $\tilde R$ is unitary since $T T^\dagger = \1_{2k}/2$, $R R^\dagger = R R^T = \1_{2k}$, therefore these matrices satisfy:
\begin{equation}
	AA\++BB\+=\1_k,\quad AB^T+BA^T=0.\label{eq:constraints}
\end{equation}
These conditions on $A$ and $B$ are also sufficient to ensure that $R = T^{-1}\tilde RT\in\O(2k, {\mb{R}})$, since
\begin{equation}
    R = 
    \begin{pmatrix}
        \Re(A + B) & \Im(B - A)\\
        \Im(A + B) & \Re(A - B)
    \end{pmatrix}.
\label{eq:RABparam}
\end{equation}

Note that we could have also parametrized the Eq.~(\ref{eq:Elexpr}) by integrating over the manifold parametrized by the more standard antisymmetric $Z$ matrices of Eq.~(\ref{eq:Zmatrix}).
In terms of these matrices, the matrix associated to the modes $\{d_m\}$ is $Z=-A^{-1}B$, because due to the transformation rules Eq.~\eqref{eq:tilderotation}, we have that:
\begin{equation}
    0=d_m\ket e \propto \big(U_Rd_m U_R\+\big)\big(e^{\frac{1}{2}d\+Zd\+}\ket e\!\big)=e^{\frac{1}{2}d\+Zd\+}\big(A_{mn}\big(d_n+ Z_{nl}d_l\+\big)+ B_{ml}d\+_l\big)\ket e.
\end{equation}
\subsection{Computation of the matrix elements}
With this parametrization, we proceed to evaluate $\Omega_0(R)$, $\Omega_e(R)$, and $\Omega_\eta(R)$.
First, we have \cite{Oi_2012,AAAA}:
\begin{equation}
	\Omega_e(R)=\bra e U_R\ket e = \sqrt{\det A}.\label{eq:simpleoverlap}
\end{equation}
For $\Omega_\eta$, we can use the same formula by writing $\ket\eta=U_{R_\eta}\ket e$.
The stabilizer conditions of Eq.~\eqref{eq:estab} imply that $\ket \eta$ is annihilated by the modes
\begin{equation}
	\Big\{\,\frac{1}{2} (d_m-d_{m+1}-d\+_m-d\+_{m+1}),\ \frac{1}{2}(d_1+d_k+d_1\+-d_k\+)\,\Big\}_{m=1}^{k-1}.
\end{equation}
Thus, since according to Eq.~\eqref{eq:tilderotation} these modes are obtained as $\tilde R_\eta \binom{\{d_m\}}{\{d_m\+\}}$, we have that
\begin{equation}
	\tilde R_\eta=\begin{pmatrix} A_\eta & B_\eta \\ B^*_\eta & A^*_\eta \end{pmatrix},\quad
	A_\eta =-\frac{1}{2}(C-\1_k) ,\quad B_\eta = -\frac{1}{2}(C+\1_k),\quad
	C\sim\begin{pmatrix}
	0 & +1 & 0 & 0  \\
	0 & 0 & +1 & 0  \\
	0 & 0 & 0 & +1  \\
	-1 & 0 & 0 & 0 
	\end{pmatrix},\label{eq:etamatrix}
\end{equation}
where $C$ is a pseudo-permutation real orthogonal matrix with the property $C^k=-\1_k$ (written above in the $k=4$ case as an example).
Then using the previous result Eq.~\eqref{eq:simpleoverlap}: 
\begin{equation}
	\Omega_\eta(R)=\bra e \big(U_{R_\eta}\+U_R\big)\ket e=\bra e U_{(R_\eta^TR)}\ket e=\sqrt{\det\big(A_\eta\+A+B_\eta^TB^*\big)}.
\end{equation}
In Eq.~\eqref{eq:etamatrix} we chose the phase of the $\tilde R_\eta$ matrix such that $\Omega_\eta(R=\1_{2k})^2=1/2^{k-1}=\braket{\1^\eta}{\1^e}$ (note that for $R=\1_{2k}$ we have $A=\1_k$ and $B=0$).
Finally, we can compute $\Omega_0$ using the same procedure, applied to a two-site system.
The stabilizer conditions of Eq.~\eqref{eq:0statestab} imply that $\ket{\phi_0}^{\ot 2k}$ is annihilated by the modes
\begin{equation}
	\Big\{\,
	\frac{1}{2}(d_{m,1} + i d_{m,2}+d_{m,1}^\dagger + i d_{m,2}^\dagger),\ 
	\frac{1}{2}(-i d_{m,1} - d_{m,2}+i d_{m,1}^\dagger + d_{m,2}^\dagger)
	\,\Big\}_{m=1}^{k}.
\end{equation}
From this we can derive the 2-site block structure for the $A_0$ and $B_0$ matrices, as in the $\eta$ case in Eq.~(\ref{eq:etamatrix}):
\begin{equation}
	A_0 = \frac{1}{2} \begin{pmatrix} \1_k & i\1_k \\ -i\1_k & -\1_k \end{pmatrix}, \qquad
	B_0 = \frac{1}{2} \begin{pmatrix} \1_k & i\1_k \\ i\1_k & \1_k \end{pmatrix}
\end{equation}
thus, if we write $U_R\ot U_R=U_{R'}$ in Eq.~\eqref{eq:definition-omega} with $R'=\mqty(R & 0 \\ 0 & R)$, we get:
\begin{equation}
	\Omega_0(R)= \bra{e,e}U_{(R'^T R_0)}\ket{e,e}= \sqrt{\det\big[\big(A\++B^T\big)\big(A\+-B^T\big)\big]} = \sqrt{\det\big[\big(A^\ast - B\big)\big(A^\ast + B\big)\big]}
\end{equation}
where we normalize the overlap such that $\Omega_0(\1_{2k})=1$, i.e., with normalization $\Vert {\ket{\phi_0}^{\ot 2k}}\Vert^2=2^k$, as is our convention.

\subsection{Saddle-point computation of the Free fermion Page curve} 
With these matrix elements, Eq.~(\ref{eq:Elexpr}) reduces to:
\begin{equation}\label{eq:integralpage}
	\overline{E_k(\ell)} =  \mc D_k\int_{\O(2k, {\mb R})}\!\!\!\!\dd R\ 
	\det (A)^{\ell}\,\det\big(A_\eta\+A+B_\eta^TB^*\big)^{L-\ell}\,\det\big[\big(A^\ast - B\big)\big(A^\ast + B\big)\big]^{L/2},
\end{equation}
where $R$ is parametrized in terms of $A$ and $B$ using Eq.~(\ref{eq:RABparam}), which satisfy the constraints of Eq.~(\ref{eq:constraints}).
We wish to evaluate this for large $L$ and $\ell = r L$ using a saddle-point approximation. 
To identify the dominating saddle, we first identify two symmetries of this integral:
\begin{enumerate}
\item[(i)]
It is clear that the integrand is a homogeneous polynomial function of the components of the matrices $A,A^*,B,B^*$.
While the matrix elements of $A$ and $B$ can be complex, the coefficients of the polynomial are real.
The integral is then symmetric under complex conjugation, where mapping $(A,B)\mapsto (A^*,B^*)$ simply results in the complex conjugate of the integrand. This also shows that $\overline{E_k(\ell)}$ is real, although we know this by construction.
\item[(ii)] 
Notice also that the constraints of Eq.~\eqref{eq:constraints} are invariant under $(A,B)\mapsto (O^T A O,O^T B O)$ where $O$ is a real orthogonal matrix.
The terms $\det(A)$ and $\det(A^\ast \pm B)$ are also invariant under this transformation, and if $[C,O]=0$, then also $\det(A\+_\eta A+B^T_\eta B^*)$ does not change.
Since $C$ is itself orthogonal, the integrand is therefore invariant under rotations by elements of the group generated by $C$, i.e., $O \in \{\1,C,C^2,...,C^{2k-1}\}$.
\end{enumerate}
Given these symmetries, we will look for a symmetric saddle point that dominates the integral for large $L$ and $\ell = r L$, which is specified by saddle-point values $A = A_{\ast}$ and $B = B_{\ast}$.
While this imposition of the symmetry onto the saddle point is just an ansatz, we find that it leads to physically sensible answers.
Imposing that $A_{\ast}$ and $B_{\ast}$ are invariant under these symmetries, we obtain that 
\begin{equation}
    (A_{\ast})^\ast = A_{\ast},\;\;\;(B_{\ast})^\ast = B_{\ast},\;\;\;[C,A_{\ast}] = [C,B_{\ast}] = 0.
\label{eq:saddleconditions}
\end{equation}
The spectrum of $C$ is non-degenerate, with eigenvalues given by the $k$-th roots of $-1$:
\begin{equation}
    C=V\+ \mrm{diag}(\{f_p\})V\;\;\;\implies\;\;\;f_p = -e^{-i \frac{2\pi}{k}p},\;\;\;p = -\frac{k-1}{2}, \cdots, \frac{k-1}{2},
\label{eq:fpdefn}
\end{equation}
where we will refer to $p$ as the ``momenta''.
Moreover, since $C$ is real we can show that $P=VV^T$ is the momentum flipping operator:
\begin{equation}
	 C = C^\ast\;\;\implies\;\;\mrm{diag}(\{f_p\}) = P\,\mrm{diag}(\{f_p\})^* P\+ = P\,\mrm{diag}(\{f_{-p}\})P\+.
\end{equation}
Since the spectrum of $C$ is non-degenerate, this means that $A_{\ast}$ and $B_{\ast}$ are simultaneously diagonalizable together with $C$, which means that
\begin{equation}
    A_\ast = V^\dagger \mrm{diag}(\{\alpha_p\})V,\;\;\;\;B_\ast = V^\dagger \mrm{diag}(\{\beta_p\})V.
\label{eq:ABdiag}
\end{equation}
Since $A_{\ast}$ and $B_{\ast}$ are also real, it must also hold that $\alpha_p=\alpha_{-p}^*$ and $\beta_p=\beta_{-p}^*$, which can be shown using the momentum flipping operator $P$.
From this, we can show that the conditions of Eq.~\eqref{eq:constraints} on this ansatz space reads:
\begin{equation}
	\forall p: \alpha_p \alpha_{-p} +\beta_p\beta_{-p} = 1,\quad
	\alpha_p \beta_{-p} +\beta_p\alpha_{-p} = 0,
\label{eq:alphabetaconst}
\end{equation}
which can be obtained by explicit substitution of Eq.~(\ref{eq:ABdiag}) and using the properties of the operator $P$.
Similarly, using Eq.~(\ref{eq:etamatrix}) and (\ref{eq:fpdefn}), the integrand of Eq.~\eqref{eq:integralpage} becomes
\begin{equation}
    \prod_{p=-(k-1)/2}^{(k-1)/2} (\alpha_p)^\ell \left(\frac{1-f_{-p}}{2}\alpha_p-\frac{1+f_{-p}}{2}\beta_{p}\right)^{L-\ell} ((\alpha_{p})^2-(\beta_{p})^2)^{L/2} {\;\;\defn e^{L\cdot S[\{\alpha_p\},\{\beta_p\}]}},
\end{equation}
where, within this diagonal ansatz space, the integrand is real and positive, and we can write it in terms of an action $S[\{\alpha_p\},\{\beta_p\}]$ defined as
\begin{equation}
	S[\{\alpha_p\},\{\beta_p\}] = -\sum_{p=-(k-1)/2}^{(k-1)/2}\left( r \log\alpha_p +
	(1-r)\log(\frac{1-f_{-p}}{2}\alpha_p-\frac{1+f_{-p}}{2}\beta_{p})+
	\frac{1}{2}\log(\alpha_{p}^2-\beta_{p}^2)\right).
\end{equation}
We then want to find the values of $\{\alpha_p\}$ and $\{\beta_p\}$ that satisfy the constraints of Eq.~(\ref{eq:alphabetaconst}) and extremize this action.
Since the constraints couple opposite modes, we can decouple the action into a sum of two-mode actions (with an extra single-mode $p=0$ action if $k$ is odd), as
\begin{equation}
    {S[\{\alpha_p\}, \{\beta_p\}] = \sum_{p = \frac{1}{2}}^{\frac{k-1}{2}}{S_{(\pm p)}}\;\;\;\text{if $k$ even},\qquad S[\{\alpha_p\}, \{\beta_p\}] = S_{(p = 0)} + \sum_{p = 1}^{\frac{k-1}{2}}{S_{(\pm p)}}\;\;\;\text{if $k$ odd}},
\end{equation}
where we have the two-mode actions
\begin{equation}
\begin{split}
	S_{(\pm p)}=-\Big( r \log(\alpha_p\alpha_{-p}) +
	(1-r)\log(\frac{1-f_{-p}}{2}\alpha_p-\frac{1+f_{-p}}{2}\beta_{p})+\qquad\qquad\qquad\qquad\qquad\qquad\qquad\\
	\qquad\qquad\qquad+(1-r)\log(\frac{1-f_{p}}{2}\alpha_{-p}-\frac{1+f_{p}}{2}\beta_{-p})+
	\frac{1}{2}\log(\alpha_{-p}^2-\beta_{-p}^2)+
	\frac{1}{2}\log(\alpha_{p}^2-\beta_{p}^2)\Big),
\end{split}
\end{equation}
and the single-mode action (which is only relevant when $k$ is odd)
\begin{equation}
    S_{(p=0)} = -\left( r \log\alpha_0 +
	(1-r)\log(2\alpha_0)+
	\frac{1}{2}\log(\alpha_{0}^2-\beta_{0}^2)\right).
\label{eq:zeromode}
\end{equation}
We can then solve for the extrema of each of these actions $S_{(p = 0)}$ and $S_{(\pm p)}$ separately for each $p$.
Starting with the two-mode action, under the given constraints of Eq.~(\ref{eq:alphabetaconst}), $\{\alpha_p\}$, $\{\beta_p\}$, and $\{f_p\}$ can be fully parametrized as
\begin{equation}
\begin{gathered}
    f_p=e^{i\omega_p},\;\;\;\omega_p = \pi - \frac{2\pi p}{k}\\
    \alpha_p=\cos\theta_p e^{i\phi_p},\;\;\;\beta_p=i\sin\theta_p e^{i\phi_p},\;\;\;\text{where}\;\;\theta_{-p}=-\theta_{p},\;\;\phi_{-p}=-\phi_{p},\label{eq:omegapdef}
\end{gathered}
\end{equation}
The two-mode action then becomes:
\begin{equation}
	S_{(\pm p)}(\theta_{p},\phi_{p})=-\left( r \log(\cos^2\theta_p) + (1-r)\log(\sin^2(\theta_p - \omega_p/2)) \right)
\end{equation}
which is fully independent of $\phi_p$, causing a saddle-point degeneracy.
The saddle point equation for the $\theta_p^\ast$ that extremizes this action is:
\begin{equation}
	r \tan(\theta_p^*) = (1-r) \cot(\theta_p^* - \omega_p/2) \ \implies\ \tan(\theta_p^*) = \frac{\tan(\omega_p/2) \pm \sqrt{\tan^2(\omega_p/2) + 4r(1-r)}}{2r}\label{eq:saddlepointeq}.
\end{equation}
For the single-mode action that appears when $k$ is odd the constraints impose $\alpha_0,\beta_0$ are both real (from complex-conjugation symmetry), and one of them is $0$ and the other has norm $1$ (from the matrix constraints).
The saddle point solution to Eq.~(\ref{eq:zeromode}) in this case is simply $\alpha_0=1$ and $\beta_0=0$, thus resulting in a zero contribution to the action $S_{(p = 0)}=0$.
In total, for both even and odd $k$, we have:
\begin{equation}
		-\frac{1}{L}\log (\overline{E_k(\ell)})\sim S_*(r)= -\sum_{p>0}^{\frac{k-1}{2}} \left( r \log(\cos^2\theta_p^*) + (1-r)\log(\sin^2(\theta_p^* - \omega_p/2)) \right)\label{eq:solution}
\end{equation}
If we specialize this to $k=2$ we get one pair of momenta $p=\pm 1$, corresponding to $f_{\pm 1}=\pm i$:
\begin{equation}
	t\equiv \tan\theta^* = \frac{1 - \sqrt{1 + 4r - 4r^2}}{2r}\quad\implies\quad S_*(r) = \log(1+t^2) - (1-r)\log\left(\frac{(1-t)^2}{2}\right)
\end{equation}
Which perfectly replicates Eq.~(49) of Ref.~\cite{swann2025}, where $q\equiv 2r-1$ and $2L\mapsto L$ (since in the reference, the system consists of $L$ Majorana modes instead of $2L$).
\subsection{Exact computation of the Free fermion Page curve for $k=2$} 
As an additional example, we can compute $\overline{E_2(\ell)}$ exactly for any $\ell$ and $L$, exploiting the exceptional decomposition $\so(4)\cong\su(2)\oplus\su(2)$.
Following Ref.~\cite{swann2025}, we can directly express $\overline{E_2(\ell)}$ as
\begin{equation}
    \overline{E_2(\ell)}=\big(\bra{\up}^{\ot 2\ell}\bra{\rt}^{\ot 2(L-\ell)}\big)\Pi_\mrm{Heis}\big(\ket{\up\up}+\ket{\dn\dn}\big)^{\ot L},
\end{equation}
where $\Pi_\mrm{Heis}$ is the projector onto the ground state space of the spin-$\frac{1}{2}$ ferromagnetic Heisenberg model $h_{ij}=(1-X_iX_j-Y_iY_j-Z_iZ_j)$, which appears within a sector of the more general $\SO(4)$ Heisenberg model discussed in Sec.~\ref{subsec:Majoranaeffectiveham} relevant for $k = 2$, and $\ket\rt =(\ket\up+\ket\dn)/\sqrt 2$.
Note the similarity with Eq.~\eqref{eq:tocompute}.
The dimension of the ground state space is $\mc D_2=2L+1$, and the associated ground state manifold is the Bloch sphere, which matches with $M_\mrm{MG}^{(2,+)}\cong \mb S^2$. The integral Eq.~\eqref{eq:coherent-state-projection} will therefore be over $\SU(2)$ spin-coherent states:
\begin{equation}
    \overline{E_2(\ell)}=(2L+1)\int_{\SU(2)}\dd U\,\big(\bra{\up}^{\ot 2\ell}\bra{\rt}^{\ot 2(L-\ell)}\big)\big(U\ketbra{\up}{\up}U\+\big)^{\ot 2L}\big(\ket{\up\up}+\ket{\dn\dn}\big)^{\ot L}.
\end{equation}
If we parametrize $U\in\SU(2)$ as:
\begin{equation}
	U = \begin{pmatrix} x & -y^* \\ y & x^* \end{pmatrix},\quad |x|^2+|y|^2=1,
\end{equation}
then the integral becomes:
\begin{equation}
	\overline{E_2(\ell)} = \frac{2L+1}{2^{L-\ell}}\int_{\mb S^3}\dd\Omega(x,y)\  x^{2\ell} (x+y)^{2(L-\ell)}((x^*)^2+(y^*)^2)^L,
\end{equation}
where $\mb S^3$ is the unit sphere in $\mb C^2$, and the measure $\dd\Omega$ on $\mb S^3$ is normalized to have volume $1$.
Notice that when expanding the binomials, due to the spherical symmetry, the only terms that survive integration are the ones where the powers of $x$ and $y$ perfectly match those of $x^*$ and $y^*$ respectively.
We finally obtain:
\begin{equation}
	\overline{E_2(\ell)} = \frac{2L+1}{2^{L-\ell}} \sum_{k=0}^{L-\ell}\binom{2(L-k)}{2k}\binom{L}{k} \int_{\mb S^3}\dd\Omega(x,y)\ |x|^{4(L-k)} |y|^{4k}=\frac{1}{2^{L-\ell}} \sum_{k=0}^{L-\ell}\frac{\binom{2(L-k)}{2k}\binom{L}{k}}{\binom{2L}{2k}},
\end{equation}
where the spherical integral can be computed by turning it into a Gaussian integral which then results in a product of Gamma functions~\cite{Folland01052001}.
This reproduces a result which was obtained in Ref.~\cite{lastres2026} by projecting onto $\mrm{Com}_{\rm MG}^{(2)}$ using an orthonormal basis.

\end{document}